%% file: main.tex
\documentclass[11pt]{article}
\RequirePackage{algorithm,algorithmic,graphicx,amssymb,epsf,epic,url,complexity}
\RequirePackage{fullpage}
\usepackage{enumitem}
\usepackage{authblk}
\usepackage{mdframed}
\usepackage{microtype, complexity}
\usepackage[usenames,dvipsnames]{xcolor}
\usepackage{amsthm}
\usepackage{amsmath}
\usepackage{upref}
\usepackage{amsfonts}
\usepackage{graphicx}

\usepackage{latexsym}
\usepackage{amssymb}
\usepackage{amscd}
\usepackage{mdframed}
\usepackage[all]{xy}

\usepackage{xcolor}
\usepackage{amsthm}
\usepackage{framed}
\usepackage[margin=1in]{geometry}

\usepackage{mathrsfs}
\usepackage{amsfonts}
\usepackage{epsfig}
\usepackage{epstopdf}
\usepackage{tcolorbox}
\usepackage{tikz}
\usepackage[hang, small,labelfont=bf,up,textfont=it,up]{caption} 
\usepackage{authblk}
\usepackage[backref=page, pagebackref=true]{hyperref}
\usepackage[capitalize]{cleveref}

\newcommand {\ignore} [1] {}
\def\eps{\varepsilon}

\colorlet{shadecolor}{gray!15}

\newtheorem{thm}{Theorem}

\newtheorem{invariant}{Invariant}

\newtheorem{cor}[thm]{Corollary}
\newtheorem{lem}[thm]{Lemma}

\newtheorem{obs}[thm]{Observation}
\newtheorem{cl}[thm]{Claim}
\theoremstyle{definition}

\theoremstyle{remark}

\input{latexmac}

\denseformat
\newtheorem{mdresult}[theorem]{Theorem}
\newenvironment{Theorem}{\begin{mdframed}[backgroundcolor=lightgray!40,topline=false,rightline=false,leftline=false,bottomline=false,innertopmargin=2pt]\begin{mdresult}}{\end{mdresult}\end{mdframed}}

\newtheorem{question}{Question}
\newtheorem{mdresult2}[question]{Question}
\newenvironment{Question}{\begin{mdframed}[backgroundcolor=lightgray!40,topline=false,rightline=false,leftline=false,bottomline=false,innertopmargin=2pt]\begin{mdresult2}}{\end{mdresult2}\end{mdframed}}

\title{Dynamic $((1+\eps)\ln n)$-Approximation Algorithms \\for Minimum Set Cover and Dominating Set \footnote{A preliminary version of this paper has appeared in the proceedings of STOC 2023.}}

\author[1]{Shay Solomon\thanks{Co-funded by the European Union (ERC, DynOpt, 101043159). Views and opinions expressed are however those of the author(s) only and do not necessarily reflect those of the European Union or the European Research Council. Neither the European Union nor the granting authority can be held responsible for them. This research was also supported by the Israel Science Foundation (ISF) grant No.1991/1, and by a grant from the United States-Israel Binational Science Foundation (BSF), Jerusalem, Israel, and the United States National Science Foundation (NSF).}}
\author[1]{Amitai Uzrad\thanks{This research was supported by the Israel Science Foundation (ISF) grant No.1991/1.}}
\affil[1]{Tel Aviv University}

\begin{document}

\date{\empty}

\begin{titlepage}
\def\thepage{}
\maketitle

\begin{abstract}
\noindent The minimum set cover (MSC) problem admits two classic algorithms: a \emph{greedy} $\ln n$-approximation and a primal-dual $f$-approximation, where $n$ is the universe size and $f$ is the maximum frequency of an element. Both algorithms are simple and efficient, and remarkably --- one cannot improve these approximations under hardness results by more than a factor of $(1+\eps)$, for any constant $\eps > 0$.

In their pioneering work, Gupta et al.\ [STOC'17] showed that the greedy algorithm can be \emph{dynamized} to achieve $O(\log n)$-approximation with update time $O(f \log n)$.
Building on this result, 
Hjuler et al.\ [STACS'18] dynamized the greedy minimum dominating set (MDS) algorithm,
achieving a similar approximation with update time $O(\Delta \log n)$ (the analog of $O(f \log n)$), albeit for unweighted instances. The approximations of both algorithms, which are the state-of-the-art, exceed the static $\ln n$-approximation by a rather large constant factor. 
In sharp contrast, the current best dynamic \emph{primal-dual} MSC algorithms  
achieve fast update times
together with an approximation that
exceeds the static $f$-approximation by a factor of (at most) $1+\epsilon$, for any $\epsilon > 0$.

This paper aims to bridge the gap between the best approximation factor of the \emph{dynamic} greedy MSC and MDS algorithms
and the \emph{static} $\ln n$ bound. We present dynamic algorithms for \emph{weighted} greedy MSC and MDS with approximation $(1+\epsilon)\ln n$ for any $\epsilon > 0$, while achieving the same update time (ignoring dependencies on $\epsilon$) of the best previous algorithms (with approximation significantly larger than $\ln n$).
Moreover, we prove that the same algorithms achieve $O(\min \{ \log n, \log C \})$ amortized {\em recourse}; the recourse measures the number of changes to the maintained structure per update step, and the cost of each set lies in the range $[\frac{1}{C},1]$.

\end{abstract}
\end{titlepage}

\pagenumbering {arabic}

\input{sec1}

\input{sec2}
\input{sec3}

\input{sec456}

\bibliographystyle{latex8}
\bibliography{randomMMbibfile}

\end{document}

%% file: latexmac.tex
\def\denseformat{
\setlength{\textheight}{9.5in}
\setlength{\textwidth}{6.9in}
\setlength{\evensidemargin}{-0.3in}
\setlength{\oddsidemargin}{-0.3in}
\setlength{\headsep}{10pt}
\setlength{\topmargin}{-0.44in}
\setlength{\columnsep}{0.375in}
\setlength{\itemsep}{0pt}
}

\newtheorem{theorem}{Theorem}[section]

\newtheorem{definition}[theorem]{Definition}

\def\boldhead#1:{\par\vskip 7pt\noindent{\bf #1:}\hskip 10pt}
\def\ithead#1:{\par\vskip 7pt\noindent{\it #1:}\hskip 10pt}
\def\ceil#1{\lceil #1\rceil}

\def\inline#1:{\par\vskip 7pt\noindent{\bf #1:}\hskip 10pt}
\def\midinline#1:{\par\noindent{\bf #1:}\hskip 10pt}
\def\dnsinline#1:{\par\vskip -7pt\noindent{\bf #1:}\hskip 10pt}
\def\ddnsinline#1:{\newline{\bf #1:}\hskip 10pt}
\def\largeinline#1:{\par\vskip 7pt\noindent{\large\bf #1:}\hskip 10pt}
\long\def\commhide #1\commhideend{}
\long\def\commtim #1\commendt{#1}
\long\def\commb #1\commbend{}
\long\def\commedit #1\commeditend{} 

\long\def\commB #1\commBend{}       

\long\def\commex #1\commexend{}     

\long\def\commsiena #1\commsienaend{}  

\long\def\commBI #1\commBIend{}  

\long\def\CProof #1\CQED{}
\def\qed{\mbox{}\hfill $\Box$\\}
\def\blackslug{\hbox{\hskip 1pt \vrule width 4pt height 8pt
    depth 1.5pt \hskip 1pt}}
\def\QED{\quad\blackslug\lower 8.5pt\null\par}

\long\def\PPP#1{\noindent{\bf Proof:}{ #1}{\quad\blackslug\lower 8.5pt\null}}

\long\def\denspar #1\densend
{#1}

\newif\ifnotesw\noteswtrue
   {\ifnotesw\marginpar[\hfill\(\top\)]{\(\top\)}\fi}%
   {\ifnotesw\marginpar[\hfill\(\bot\)]{\(\bot\)}\fi}

\newcommand{\mnote}[1]%
    {\ifnotesw\marginpar%
        [{\scriptsize\it\begin{minipage}[t]{\marginparwidth}
        \raggedleft#1%
                        \end{minipage}}]%
        {\scriptsize\it\begin{minipage}[t]{\marginparwidth}
        \raggedright#1%
                        \end{minipage}}%
    \fi}

\def\MathF{\hbox{\rm I\kern-2pt F}}
\def\MathP{\hbox{\rm I\kern-2pt P}}
\def\MathR{\hbox{\rm I\kern-2pt R}}
\def\MathZ{\hbox{\sf Z\kern-4pt Z}}
\def\MathN{\hbox{\rm I\kern-2pt I\kern-3.1pt N}}
\def\MathC{\hbox{\rm \kern0.7pt\raise0.8pt\hbox{\footnotesize I}
\kern-4.2pt C}}
\def\MathQ{\hbox{\rm I\kern-6pt Q}}

\newsavebox{\ttop}\newsavebox{\bbot}

\def\eps{\epsilon}



%% file: sec1.tex
\newcommand{\DD}{\ensuremath{\mathcal{D}}}
\newcommand{\addrandom}{\ensuremath{\textnormal{\textsf{insert-hyperedge}}}}
\newcommand{\handleedge}{\ensuremath{\textnormal{\textsf{handle-hyperedge}}}}
\newcommand{\FF}{\mathcal{F}}
\renewcommand{\SS}{\mathcal{S}}

\section{Introduction}\label{sec:intro}

\noindent This paper studies the two closely-related problems of weighted {\em minimum set cover} (shortly, MSC or SC) and weighted {\em minimum dominating set} (shortly, MDS or DS) in the standard {\em dynamic} setting. In the static weighted SC problem, we are given a universe $\mathcal{U}$ of $n$ elements 
and a collection of $m$ sets $\mathcal{S} = (S_1,\ldots,S_m)$ whose union equals $\mathcal{U}$, where each set has a \emph{weight} (or \emph{cost}) assigned to it, and the cost of each set lies in the range $[\frac{1}{C},1]$. The goal is to find a sub-collection $\mathcal{S}'$ of $\mathcal{S}$ 
of minimum total weight, such that each element of $\mathcal{U}$ belongs to at least one set in $\mathcal{S}'$. 
As common, we denote by $f$ the maximum {\em frequency} of any element in $\mathcal{U}$, where an element's frequency is the number of sets in $\mathcal{S}$ it belongs to. 
In the static weighted DS problem, we are given an $n$-vertex $m$-edge graph $G = (V,E)$, where each vertex has a weight assigned to it. The goal is to find a subset of vertices $V' \subseteq V$ of minimum total weight, such that for any vertex $v \in V$, either $v \in V'$ or $v$ has a neighbor in $V'$. 

The SC problem is a central NP-hard problem, which admits two classic algorithms: a {\em greedy} $\ln n$-approximation and a primal-dual $f$-approximation.
One cannot achieve approximation $(1-\epsilon) \ln n$ unless P = NP \cite{williamson2011design,dinur2014analytical};
similarly, one cannot achieve approximation $f-\epsilon$ for any fixed $f$ under the unique games conjecture \cite{KR08}.
The DS problem is a special case of the SC problem, 
but it is also equivalent to it in some sense,  
as the SC problem can be efficiently reduced in an approximation-preserving fashion to MDS, thus in particular one cannot achieve for it approximation $(1-\epsilon) \ln n$ unless P = NP. 
The greedy and primal-dual approximation algorithms for set cover have been extremely well-studied in the static setting, and it is fair to say that they are pretty well understood by now.
The harmonic number $(\ln n)$ greedy approximation \cite{Johnson74,Lovasz75,Chvatal79} is a pearl of algorithms; whether or not one can beat the $\ln n$ approximation guarantee is a question that led to numerous breakthrough lower bound works (see, e.g., \cite{Feige,Slavik,Lund,GuptaNew}), where the ``only goal" of some of them was to improve small leading constants of the $\ln n$ bound.  

We note that the greedy algorithm for DS works in the same way as greedy SC: They are basically two incarnations of the same meta-algorithm, and it achieves the same approximation factor of $\ln n$. 
It is only natural to ask whether one can generalize these classic textbook algorithms for the {\em standard} dynamic setting,
where in the SC (respectively, DS) problem, at each update step a single element (resp., edge) is either inserted to or deleted from the universe (resp., graph),
where we start from a universe (resp., graph) with no elements (resp., edges), and the collection $\SS$ of sets (resp., vertex set $V$) is fixed. \footnote{Note that in the SC and DS problems, the roles of $n$ and $m$ swap: While elements in the SC problem basically correspond to edges in DS, $n$ denotes the number of elements 
in SC and $m$ denotes the number of edges in DS.} 
The dynamic setting is not only important on its own right, but may also help in understanding more thoroughly basic properties of the classic static algorithms, and in particular their robustness: {\bf How {\em robust} are the static algorithms against small adversarial input perturbations that occur repeatedly over time?}

There is a growing interest on the SC and DS problems in the dynamic setting in recent years,
where the holy grail is to coincide with the approximation factor of the static algorithm together with a low update time. 
In their pioneering work, Gupta et al.\ \cite{GKKP17} showed that the greedy SC algorithm can be {\em dynamized} to achieve $O(\log n)$-approximation with update time $O(f \log n)$. 
Building on this result, 
Hjuler et al.\ \cite{Hjuler} dynamized the greedy DS algorithm,
achieving a similar approximation with update time $O(\Delta \log n)$ (the analog of $O(f \log n)$), but only for unweighted instances;
here and throughout $\Delta$ stands for the maximum degree in the graph. 
The approximations of both algorithms, which are the state-of-the-art, exceed the static $\ln n$-approximation by a factor larger than 1000.
In sharp contrast, the current best dynamic {\em primal-dual} SC algorithms provide approximations that match or nearly match the static $f$-approximation. The first primal-dual algorithms, by Gupta et al.\ \cite{GKKP17} and independently by Bhattacharya et al.\ \cite{BCH17},
provided an $O(f^3)$-approximation with $O(f^2)$ update time.
These results were subsequently improved by Abboud~et al.\ \cite{abboud} and Bhattacharya et al.\ \cite{bhattacharya2019new}, which achieve update times of $O(\frac{f^2}{\epsilon^5}\log n)$ and $O(\frac{f}{\epsilon^2}\log (Cn))$, respectively, with an approximation of $(1+\epsilon)f$ (for any $\epsilon > 0$). These results were further improved by Bhattacharya et al.\ \cite{Bhattacharya} and Assadi-Solomon \cite{Assadi}, which achieve update times of $O(\frac{f^2}{\epsilon^3} + \frac{f}{\epsilon^2}\log C)$ and $O(f^2)$, respectively, with an approximation of $(1+\epsilon)f$ (for any $\epsilon > 0$) and $f$, respectively. We note that the algorithm of Bhattacharya et al.\ \cite{Bhattacharya} is deterministic and applies to weighted instances, while that of Assadi-Solomon \cite{Assadi} is randomized and works against an oblivious adversary, and moreover it applies only to unweighted instances, but its approximation is $f$ rather than $(1+\epsilon)f$. 
Recently, Bukov et al.\ \cite{bukov2023nearly}  presented a deterministic algorithm and a randomized algorithm that works against an adaptive adversary, both achieving a $((1+\eps)f)$-approximation, with  amortized update times of $O\left(\frac{1}{\epsilon}f\log f + \frac{f}{\epsilon^3} + \frac{f\log C}{\epsilon^2}\right)$ and $O\left(\frac{f}{\epsilon^3}\log^*f  + \frac{f\log C}{\epsilon^3}\right)$, respectively.

\begin{Question} \label{q1}
Is it possible to achieve approximation better than $c \ln n, c > 1000$ (as in \cite{GKKP17, Hjuler}) for dynamic SC and DS,
even for \emph{unweighted} instances, and ideally an approximation approaching $\ln n$, with a low update time?
\end{Question}

Another important quality measure of dynamic algorithms besides their approximation guarantee and update time is the {\em recourse}, which measures the number of changes to the maintained structure per update step; one can optimize the amortized as well as the worst-case recourse bounds, just as with the update time measure. There is a large body of work from recent years on dynamic algorithms with low recourse; see, e.g., \cite{Bernstein, roie,solo}, and the references therein. In the same paper \cite{GKKP17}, Gupta et al.\ presented an $O(\log n)$-approximation algorithm (still for approximation $c \ln n$, for $c > 1000)$ with constant amortized recourse and polynomial update time. 
Gupta and Levin \cite{roie} generalized the results of Gupta et al.\ on set cover 
with bounded amortized recourse for fully dynamic submodular covering problems by introducing a general framework. The framework of \cite{roie} applies to a very general setting, but it is not very relevant to our work, as it does not bound the update time, and it does not lead to an approximation ratio of $(1+\epsilon)\ln $.

A trivial approach to achieve an approximation ratio of $(1+O(\epsilon))$ for \emph{unweighted} SC is to compute the SC, denote its size by $OPT$, and then after $\epsilon \cdot OPT$ update steps recompute it, and repeat in this way. Although this approach yields an amortized recourse bound of $O(\frac{1}{\epsilon})$, it has exponential running time, far from what we are aiming for. An alternative approach, still for \emph{unweighted} SC, would be to compute an approximation for the SC using the greedy algorithm, and then recompute after $\mathtt{poly}(\epsilon) \cdot Greedy$ update steps (here $Greedy$ denotes the obtained SC size), and repeat; a priori it might not be clear, but we can show that this algorithm yields a $(1+O(\epsilon)) \cdot \ln n$ approximation factor. 
However, the update time is at least linear in $n$, still far from what we are aiming for. Even for \emph{unweighted} instances, it is unclear how to achieve a $(1+\epsilon) \cdot \ln n$ approximation factor with low update time and low recourse.

\begin{Question} \label{q2}
Is it possible to achieve approximation better than $c \ln n, c > 1000$ (as in \cite{GKKP17, Hjuler}) for dynamic \emph{unweighted} SC and DS, 
and ideally an approximation approaching $\ln n$, with a low update time \emph{and} constant recourse?
\end{Question}

For \emph{weighted} SC, it is not clear whether one can achieve an approximation ratio approaching $\ln n$ with low update time, let alone with decent recourse. Recall that in the \emph{weighted} setting, we assume that the cost of each set lies in the range $[\frac{1}{C},1]$. If the answer to Question 2 is affirmative, one might be able to naturally generalize for \emph{weighted} instances by blowing up the recourse by a factor of $C$, which is not ideal. 

\begin{Question} \label{q3}
Is it possible to achieve approximation better than $c \ln n, c > 1000$ (as in \cite{GKKP17, Hjuler}) for dynamic \emph{weighted} SC and DS, 
and ideally an approximation approaching $\ln n$, with a low update time \emph{and} low recourse?
\end{Question}

\subsection{Our Contribution} \label{cont} 
We present a dynamic algorithm for {\em weighted} greedy SC and DS with approximation $(1+\epsilon)\ln n$ for any $\epsilon > 0$, while achieving the same update time (ignoring dependencies on $\epsilon$) of the best previous algorithms (with approximation $c \ln n$, for $c > 1000$), and with amortized recourse of $O_{\epsilon}(\min\{\log n, \log C\})$. 
This answers Questions 1, 2 and 3 in the affirmative ($C=1$ in the \emph{unweighted} setting). We summarize the result in the following statement.

\begin{Theorem} \label{tm1}
One can dynamically maintain a $((1+\epsilon)\ln n)$-approximate \emph{weighted} SC and DS with update time $O(\mathtt{poly}(\frac{1}{\epsilon}) f \log n)$ and $O(\mathtt{poly}(\frac{1}{\epsilon}) \Delta \log n)$, respectively, and with an amortized recourse of $O(\mathtt{poly}(\frac{1}{\epsilon})\min\{\log n, \log C\})$ for any $\epsilon > 0$. 
\end{Theorem}

\noindent \textbf{Remark.} In fact, the approximation factor is $((1+\epsilon)\ln \kappa)$, where the size of each set is bounded by $\kappa$. In the DS case, each vertex can dominate up to $\Delta+1$ vertices (its up to $\Delta$ neighbors and itself), thus we get a $((1+\epsilon)\ln (\Delta+1))$ approximation factor. 

\noindent {\bf Remark.} The exact $\epsilon$-dependence in the update time bounds of Theorem \ref{tm1} is $\epsilon^{-5}$, but we did not try to optimize it, aiming for simplicity. It can be reduced to $\epsilon^{-4}$ using a more careful argument. The exact $\epsilon$-dependence in the recourse bound of Theorem \ref{tm1} is $\epsilon^{-4}$.

\subsection{Organization} \label{prel}
In Section \ref{tech} we give a technical overview of the paper, and in Section \ref{s:basic} we present our dynamic algorithms that prove Theorem \ref{tm1}. Section \ref{time} is devoted to the update time analysis, in Section \ref{rcrs} we prove the amortized recourse, and in Section \ref{approx} we prove the approximation factor, all given in Theorem \ref{tm1}.

%% file: sec2.tex
\section{Preliminaries and Technical Highlights} \label{tech}

\subsection{Preliminaries}

Instead of discussing both DS and SC problems, for the sake of brevity, we shall limit the discussion of this section to the SC problem. Much of the preliminaries and terminology provided next follows 
\cite{GKKP17}. Regarding DS, \cite{Hjuler} adapted the scheme of \cite{GKKP17} for weighted SC to unweighted DS. 

The input is a set system $(\mathcal{U},\mathcal{F})$. 
In dynamic set cover, the input sequence is 
$\bar{\sigma} = \langle \sigma_1, \sigma_2, \ldots \rangle$,
 where request $\sigma_t$ is either $(e_t,+)$ or $(e_t,-)$. 
Denote by $A_t$ the set of \emph{active} elements at time $t$. 
The initial active set of elements is $A_0 = \emptyset$. If $\sigma_t = (e_t,+)$, then $A_t \leftarrow A_{t-1} \cup \{e_t\}$. If $\sigma_t = (e_t,-)$, then $A_t \leftarrow A_{t-1} \setminus \{e_t\}$. 
We need to maintain a feasible set cover $\mathcal{S}_t \subseteq \mathcal{F}$, i.e., the sets in $\mathcal{S}_t$ must cover the set of active elements $A_t$. 
Define $n_t = |A_t|$ and $n = \max_t (n_t)$. The frequency of an element $e \in \mathcal{U}$, denoted by $f_e$, is the number of sets of $\mathcal{F}$ it belongs to; let $f_t = \max_{e \in A_t} f_e$. 
Let $OPT_t$ be the size of the optimal set cover for the set system $(A_t, \mathcal{F})$.
Throughout we shall try to omit the subscript $t$ of the update step when it is clear from the context, to avoid cluttered notation.

For a set $S$ containing elements that it can cover, let $c(S)$ denote the \emph{cost} (or \emph{weight}) of the set. The collection of sets is fixed, and so is the cost of each set, only the elements in each set can dynamically change. Our dynamic set cover algorithm will maintain an {\em extended solution} for the problem at time step $t$, denoted by $\Phi_t$, where an element of $\Phi_t$ is a pair composed of:

\begin{enumerate}
	\item A covering set $S$.
	\item A set of elements $Cov(S)\subseteq S$, which contains the elements that are \emph{covered} by $S$. 
	
\end{enumerate}

\noindent $|Cov(S)|$ denotes the \emph{cardinality} of $Cov(S)$, and we shall also refer to it as the cardinality of the pair. The ratio $\frac{c(S)}{|Cov(S)|}$ gives the average cost per newly covered element, and one wants to minimize it. We consider the \emph{inverse ratio} $\frac{|Cov(S)|}{c(S)}$, and for technical convenience we chose to use the inverse ratio, even though this is not what was done in \cite{GKKP17}. Any element $(S, Cov(S))$ of $\Phi_t$ will be called a \emph{covering pair}. Each activated element in the universe is in exactly one \emph{Cov set}, i.e., in a set $Cov(S)$ of exactly one covering set $S$, and each set can appear only once (or zero times) as the covering set of a pair. The set cover $\mathcal{S}_t$ induced by $\Phi_t$ is composed of all sets that appear as covering of a pair in $\Phi_t$; thus $\mathcal{S}_t$ denotes the set cover of the active elements $A_t$ at update step $t$, for any $t = 0,1,2,\ldots$, and is a valid set cover. Since $A_0$ contains no elements, $\mathcal{S}_0 = \emptyset$.
The covering pairs are placed into \emph{levels} according to their \emph{inverse ratio}: the level $l$ is defined by a range $R_l := [\beta ^l , \beta ^{l+1})$, where $l \ge 0$ and $\beta = (1+\epsilon)$ for $0<\epsilon<0.4$. 
Recall that we assume that the cost of each set lies in the range $[\frac{1}{C},1]$. Thus, since the costs of sets are at most one, the inverse ratios are no smaller than $1$, and only the non-negative levels are relevant for sets in the cover.

\begin{definition} {\ }
\begin{itemize}
\item
If a pair $(S,Cov(S))$ is placed at level $l$, the elements of $Cov(S)$ are said to be covered at level $l$, and $S$ is said to be covering at level $l$. 
\item The level of a covering pair is denoted by $l_{cov}(S)$, where $S$ is the covering set and we have $l_{cov}(S) \geq 0$. If $S$ is not covering, we define $l_{cov}(S)=-1$. 
\item We denote by $l(e)$ the level in which element $e$ is covered at.
\end{itemize}
\end{definition}

\subsection{Technical Highlights - Update Time}
In this high-level technical overview for the update time analysis we mostly focus on the unweighted setting, since it captures many of the challenges that we faced. Of course, the adaptation to the weighted setting also requires some additional efforts, but we omit the details in this discussion. 

\noindent \textbf{Comparison to Previous Work:} In the dynamic partition of \cite{GKKP17}, the range, which we denote here by $R'_l$, is 
$[2^{l},2^{l+10}]$;
our range assignment is thus different from  \cite{Hjuler,GKKP17} in two ways:
First, in our range assignment consecutive ranges differ by a factor of $\beta = 1+\epsilon$ rather than 2 as in \cite{Hjuler,GKKP17};
using a growth rate of $1+\epsilon$ as we do is crucial for achieving the desired approximation of $(1+\epsilon)\ln n$, but it naturally poses some nontrivial challenges. 
Second, in our range assignment the ranges $R_l$ and $R_{l'}$ are disjoint for any distinct levels $l$ and $l'$,
whereas consecutive ranges in the assignment of \cite{Hjuler,GKKP17} overlap to within a factor of $2^{10}$.
This $2^{10}$ factor overlap is crucially used by the arguments of \cite{Hjuler,GKKP17}, as explained below,
to achieve a low update time, namely update times of $O(f \log n)$ and $O(\Delta \log n)$ in \cite{GKKP17} and \cite{Hjuler}, respectively. On the negative side, the approximation guarantee in the
algorithms of \cite{Hjuler, GKKP17} is blown up by this factor of $2^{10}$; to be more precise, the approximation there is at least $2^{10} \log_2 n = (2^{10} \log_2 e) \ln n > 1477 \ln n$. As we explain below, while the leading factor ($\sim 1477$) in the approximation guarantee was not fully optimized in the works of \cite{Hjuler,GKKP17}, it cannot be reduced all the way to $(1+\epsilon)$ or even close to that.

Ideally, each pair $(S,Cov(S))$ will be placed at a level $l, 0 \le l \le \log_{\beta}(n \cdot C)$, such that $\frac{|Cov(S)|}{c(S)} \in [\beta ^{l} , \beta ^{l+1})$.
The natural analog for this in \cite{GKKP17} (and similarly in \cite{Hjuler}) is that
$\frac{|Cov(S)|}{c(S)} \in [2^{l},2^{l+10}]$, and such a set $S$ would be called \emph{clean}.
A key invariant maintained by the algorithms of \cite{Hjuler,GKKP17} is:

\vspace{1mm}
\noindent {\emph{Basic Invariant 1:} (\cite{Hjuler,GKKP17}) All pairs $(S,Cov(S))$ are always clean.
\vspace{0.1mm}

\noindent Another invariant maintained by the algorithms of \cite{Hjuler,GKKP17} is that elements are covered at the highest possible level with respect to the range assignment. To be more precise, we introduce the following definition.
\begin{definition} \label{nj}
Define $N^j(S)$ to be the set of elements $e$ that can be covered by $S$ and are covered currently at level $j$. That is to say, $N^j(S)=\{e \in S : l(e)=j\}$. 
\end{definition}

\noindent Let $S$ be any (not necessarily covering) set. 
Then set $S$ is called {\em positive clean} (in short, PC) if for every $j$, $\frac{|N^j(S)|}{c(S)} \leq 2^{j+10}$
(in the algorithms of \cite{GKKP17} and \cite{Hjuler}).

\vspace{1mm}
\noindent \emph{Basic Invariant 2:} (\cite{Hjuler,GKKP17}) All sets are positive clean at all times.
\vspace{1mm}

\noindent Maintaining these two basic invariants can be carried out in a rather straightforward way, thus leading to the neat algorithms
of \cite{Hjuler,GKKP17}. A ``clean up''   occurs for every pair $(S,Cov(S))$ that ceases to be clean by basically letting it rise to a higher level or fall to a lower one, based on the inverse ratio of $S$, $\frac{|Cov(S)|}{c(S)}$.
A similar clean-up operation occurs whenever a set $S$ ceases to be positive clean, where the level $j$ for which $S$ is not positive clean is computed, 
and the pair $(S, Cov(S))$, for $Cov(S) = N^j(S)$, is  created at a level such that the pair would be clean. 
Thus, beyond some standard bookkeeping work, the update algorithm of \cite{GKKP17} (and similarly in \cite{Hjuler}) is comprised of a (possibly empty) cascade of clean-up operations, which ends whenever all pairs $(S,Cov(S))$ are clean and all sets $S$ are positive clean. 
While it is a priori unclear why the cascade of clean-up operations terminates, let alone yields a fast algorithm,
its efficiency stems from a clever amortized analysis, which hinges on the aforementioned $2^{10}$ factor overlap --- that blows up the approximation guarantee. In more detail, the cascade of clean-up operations triggers rises and falls of covering sets and covered elements.
The main lemma in the amortized analysis of \cite{Hjuler,GKKP17} asserts that the amortized number of level changes of covered elements is 
linear in the number of levels $\log_2 n$ in the hierarchy, namely $O(\log n)$.
At the outset there are no elements and none of the sets are in the set cover.
Each element starts off with $\Theta(\log n)$ fresh tokens in its disposal, which allows it to rise through all the levels of the hierarchy once. 
However, due to element falls, elements may perform arbitrarily many rises. 
The key invariant in the amortized analysis of \cite{Hjuler,GKKP17} is that an element that is covered at level $l$, for any $l$, has $\Theta(\log n - l)$ tokens at its disposal.
Whenever any element $e$ rises, by leaving some $Cov(S)$ set and joining another one at a higher level, 
it leaves behind $O(1)$ tokens to its old set $Cov(S)$. Clearly, the invariant would hold if there were no falls.
The reason falls do not violate the validity of the invariant is due to the aforementioned $2^{10}$ factor overlap --- a fall of a pair $(S,Cov(S))$ occurs only after the {\em great majority} (roughly a $\left(1 - \frac{1}{2^{10}}\right)$-fraction) of its elements have left it to higher Cov sets. 
Since each of the elements that have left $Cov(S)$ gave $O(1)$ tokens to it, this surplus of tokens can be given to the remaining few elements (roughly a $\frac{1}{2^{10}}$-fraction) in it;
a direct calculation in \cite{Hjuler,GKKP17} shows that if $(S, Cov(S))$ falls
down $i$ levels, each of the remaining vertices in $Cov(S)$ will be given at least $i$ fresh tokens, which suffices for satisfying the invariant at its new level.
{\bf While it appears that the constant factor overlap of $2^{10}$ can be reduced to some extent, it must be \emph{inherently larger} than $1+\epsilon$ {\em for any charging argument to hold}}; indeed, the constant must be large enough so that the fraction of remaining elements in a fallen pair will be sufficiently smaller than the fraction of elements that left the pair.

\bigskip

\noindent \textbf{Our Work:} In contrast to \cite{Hjuler,GKKP17}, our goal is to achieve a near-optimal approximation factor of $(1+\epsilon) \ln n$; 
to achieve such an approximation without blowing up the update time, it is impossible to use the
aforementioned $2^{10}$ factor overlap. 
As mentioned, if one were to tweak the algorithms of \cite{Hjuler,GKKP17} by using a $(1+\epsilon)$ factor overlap instead of
$2^{10}$, then
the amortized analysis of \cite{Hjuler,GKKP17} would break. 
Instead, in our range assignment, the ranges are disjoint and they have a growth rate of $(1+\epsilon)$.
To achieve a fast algorithm, we do not aim at maintaining the two aforementioned invariants (as done in \cite{Hjuler,GKKP17}),
but rather allow violations up to a certain extent; thus our algorithm is more ``lazy'' than those of \cite{Hjuler,GKKP17}.
The exact invariants that our algorithm maintains are presented in 
Section \ref{s:basic} (see Invariants \ref{inv1}-\ref{inv3}). 
First we relax {\em Basic Invariant 2} in the following way. Define $N_j(S)$ to be the set of elements $e$ that can be covered by $S$ and are covered currently at a level \emph{lower} than $j$. That is to say, $N_j(S)=\{e \in S : l(e)<j\}$. We will allow $|N_j(S)|$ for each set $S$ to grow slightly beyond $c(S) \cdot \beta^{j}$, by a factor of $\beta$ (see Invariant \ref{inv1}). Once $|N_j(S)|$ surpasses $c(S) \cdot \beta^{j+1}$, we perform a \emph{local rise}, which means creating $Cov(S)$ at level $j+1$ with $N_{j+1}(S)$ as $Cov(S)$. 
While the relaxation to this invariant is ``local'' in the sense that it applies to each set separately, we next relax {\em Basic Invariant 1} 
in a ``global manner'', by roughly allowing up to an $\epsilon$-fraction 
of the \emph{weight of elements} to violate it, where the weight of an element $e$ is $\frac{1}{\beta^j}$, where $j = l(e)$ (for the weighted and unweighted cases). Whenever more than an $\epsilon$-fraction of the weight of elements violate the invariant, the approximation factor may become too poor, hence we perform a {\em partial reset} up to a certain level which we will call \emph{critical} (see Definition \ref{whatisdirty}, following \cite{abboud}), which cleans-up all pairs and sets up to that level. 
The approach of employing periodic resets is standard in dynamic algorithms, and it was used also in the context of
the set cover problem \cite{abboud,bhattacharya2019new,Bhattacharya} for approximation $(1+\epsilon)f$.
Indeed, our approach builds on techniques established in previous works: \cite{abboud,bhattacharya2019new,Bhattacharya} for $(1+\epsilon)f$-approximation
and \cite{Hjuler,GKKP17} for $c \ln n$-approximation for $c > 1000$.
However, since we aim for approximation $(1+\epsilon)\ln n$ rather than $(1+\epsilon)f$ or $c \ln n$ for $c > 1000$, our algorithm deviates significantly from previous work in several crucial places. 
Moreover, the analysis of our algorithm for the update time
introduces several new insights
(refer to Section \ref{time}  for further details).  

Finally, the maintenance of $N_j(S)$ for each set $S$ and level $j$ differs from \cite{Hjuler,GKKP17}, where they maintain $N^j(S)$. It is inspired by the works of \cite{BGS11,Sol16} on dynamic maximal matching.
Specifically, $N_j(S)$ in our definition is the set of elements that can be covered by $S$, which are currently covered at any level \emph{up to} $j-1$, whereas in \cite{Hjuler,GKKP17} 
it is the set of elements that can be covered by $S$, which are currently covered at level \emph{exactly} $j$. While this change appears to be crucial for achieving the desired approximation guarantee of $(1+\epsilon) \cdot \ln n$, it poses additional technical challenges.
More specifically, maintaining the set of elements that can be covered by $S$ which are covered at level exactly $j$ can be done efficiently in the obvious way, 
whereas maintaining the set of elements that can be covered by $S$ which are covered at any level up to $j-1$ (for all sets $S$ and all levels $j$) will naively blow up the update time by a factor of $\log (n)$. Achieving the same update times as in \cite{GKKP17} and \cite{Hjuler}, namely $O(f \log n)$ and $O(\Delta\log n)$ instead of $O(f \log^2 n)$ and $O(\Delta\log^2 n)$, respectively, requires new ideas; 
refer to Section \ref{updatetime} for further detail. 

\subsection{Technical Highlights - Recourse}

\noindent In contrast to \cite{GKKP17}, where they present two separate algorithms, one achieving a low update time and another achieving a low recourse, our goal is to design a single algorithm that achieves both. Proving that the amortized recourse of our algorithm is constant for the \emph{unweighted} setting is a non-trivial task in and of itself. We first observe that if a local rise of a set $S$ to level $j$ occurs due to Invariant \ref{inv1} (where recall that we would take $N_j(S)$ as the new $Cov(S)$) then the elements in the new $Cov(S)$ cannot all arrive from very low levels (much lower than $j$), otherwise $|N_{j'}(S)|$ would be much larger than $\beta^{j'+1}$ for some $j'<j$ (violating Invariant \ref{inv1}). Using this observation, we bound the amortized number of times elements leave their Cov sets at any fixed level $i$ due to local rises by $O_{\epsilon}(1)$. Then, we employ this and use a non-trivial charging argument, to claim that the amortized number of times a set participates in a partial reset as a set in the set cover is $O_{\epsilon}(1)$. Since the only way for a set to leave the set cover is via partial reset, we obtain the amortized recourse bound of $O_{\epsilon}(1)$.

Upon generalizing the recourse bound from the \emph{unweighted} setting to the \emph{weighted} setting, one can do so quite naturally but by blowing it up by a factor of $C$. One of the main challenges in the weighted case is that the size of Cov sets within a given level could vary by a factor of up to $C$. This in turn can cause up to $C$ times more sets participating in partial resets, which can blow up the recourse bound by a factor of $C$. Thus, in order to achieve a $\log C$ blow-up factor, a more intricate analysis is needed, as explained in detail in Section \ref{rcrs}. 

\subsection{Set Cover vs. Dominating Set}

Although our claims are given in SC terminology for the sake of brevity, they can easily and readily be adapted to the DS problem. Following certain claims where the transition to DS might not be as trivial, we will add a remark regarding the adjustments necessary for the DS problem. In dynamic DS, edges are inserted or deleted, and we maintain a feasible set of vertices $D \subseteq V$ called a \emph{dominating set} such that the union of the sets of neighbors of each vertex in $D$ including the vertices in $D$ is equal to $V$. We denote by $\Delta$ the maximum degree of the dynamic graph throughout the update sequence (the analog of $f$). Each vertex $v$ has a cost (or weight), denoted $c(v)$, where $\frac{1}{C} \leq c(v) \leq 1$ for any $v \in V$. Similarly to SC, our dynamic DS algorithm will maintain an extended solution for the problem at time step $t$, denoted by $\Phi_t$, where an element of $\Phi_t$ is a pair composed of a vertex $v$, and a set of vertices $Dom(v) \subseteq N(v)$, where $N(v)$ is the set of neighbors of $v$ including $v$ itself. The rest of the definitions presented previously in this section apply analogously to DS. In particular, the level of a covering pair is denoted by $l_{dom}(v)$, where $v$ is the dominating vertex and we have $l_{dom}(v) \geq 0$. If $v$ is not dominating, $l_{dom}(v) = -1$, and we denote by $l(v)$ the level in which vertex $v$ is dominated at. There are a few main differences worth noting between the dynamic SC and DS problems:

\begin{itemize}
\item In the DS problem, upon an edge deletion the two endpoints of the edge remain in the graph, whereas following an element deletion in the SC problem it is simply deactivated. This difference makes deletions more interesting in DS, if one endpoint dominated the other of a deleted edge.
\item In the DS problem, since $\Delta$ is an upper bound to the degree, each vertex can be dominated by up to $\Delta+1$ vertices (including itself), and each vertex can dominate up to $\Delta+1$ vertices (including itself). In contrast, in the SC problem an element can indeed be covered by only up to $f$ sets (the analog of $\Delta$), but a set $S$ can cover $\Theta(n)$ elements. Thus, due to any dominating level change of a vertex $v$ in the DS problem ($l_{dom}(v)$), we can spend $O(\Delta)$ time to update all neighbors of $v$ of this change. However, in SC, upon a level change of a set we cannot hope to notify all elements in it, since this could take $\Theta(n)$ time.
\item In the DS problem, each vertex has two roles: one as a dominating vertex and another as a dominated vertex. In SC, the sets and elements each have only one role.
\item In DS, following each update step, the \emph{two} endpoints of the edge are affected, whereas in SC only the updated element is affected. 
\item In the SC problem, the sets that can cover some element $e$ are static. Meaning, while $e$ is active, the same (up to $f$) sets can cover it. In contrast, in the DS problem edges are inserted and deleted, which means that the set of vertices that can dominate a given vertex is dynamic. Still the number of vertices that can dominate a vertex $v$ is bounded by $\Delta+1$, but these vertices can change dynamically.
\end{itemize}

%% file: sec3.tex
\section{Algorithm Description} \label{s:basic}

\subsection{Detailed Preliminaries and Invariants} \label{prel}

We begin by repeating some definitions from Section \ref{tech} in more detail:

\begin{definition} \label{nj}
Define $N_j(S)$ to be the set of elements $e$ that are contained by $S$ and are covered at a level less than $j$. That is to say, $N_j(S)=\{e \in S : l(e)<j\}$. 
\end{definition}

\begin{definition} \label{positive}
Let $S$ be any (not necessarily covering) set. 
\begin{itemize}
\item $S$ is called {\em positive clean} (in short, PC) if
for every $j > l_{cov}(S)$, $\frac{|N_j(S)|}{c(S)}<\beta^{j}$.
\item $S$  
is called {\emph{positive dirty} (in short, PD)
if there exists a level $j > l_{cov}(S)$, such that $\frac{|N_j(S)|}{c(S)}\geq \beta^{j+1}$; in more detail, we say that $S$ is \emph{positive dirty with respect to level $j$} (in short, $j$-PD), if $j > l_{cov}(S)$ and $\frac{|N_j(S)|}{c(S)}\geq \beta^{j+1}$.}
(Set $S$ may be $j$-PD w.r.t.\ multiple levels $j, j > l_{cov}(S)$.)
\end{itemize}
\end{definition}
\noindent Our algorithm will guarantee that there are no PD sets, i.e.,  it will
maintain the following invariant.
\begin{invariant} [no PD sets] \label{inv1}
For every set $S$ and every $j > l_{cov}(S)$, $\frac{|N_j(S)|}{c(S)}<\beta^{j+1}$.
\end{invariant}
\noindent Note that if a set is not PC, it is not necessarily PD, and vice versa, if it is not PD, it is not necessarily PC --- due to the 
buffer of factor $\beta$ between the two definitions.

\begin{definition} \label{negative}
Let $S$ be a covering set.
$S$ is called {\em negative clean} (in short, NC) if $\frac{|Cov(S)|}{c(S)}\geq \beta^{l_{cov}(S)}$.
\end{definition}

\begin{definition} {\ } \label{dirtycounter}
\begin{itemize}
\item For each level $j$, we define a \emph{dirt counter}, denoted by $\mathcal{D}_j$. Every time an element from the \emph{initial} Cov set (when last created) of a set $S$ at level $j$ leaves $Cov(S)$, $\mathcal{D}_j$ is incremented by $\frac{1}{\beta^j}$.
\item For each pair of levels $j$ and $i$ (where $j \leq i$) we define $\mathcal{D}_{j}^{i} := \sum_{q=j}^{i} \mathcal{D}_q$.
\end{itemize}
\end{definition}

\begin{definition} {\ } \label{whatisdirty}
\begin{itemize}
\item Denote by $\mathcal{S}_{j}^i$ the collection of sets that cover at levels between $j$ and $i$ (including both, where $j \leq i$), and denote by $\mathcal{C}_{j}^i$ their total cost. Let $\mathcal{S}_i = \mathcal{S}_{i}^i$ and $\mathcal{C}_i = \mathcal{C}_{i}^i$ 
\item A level $j$ will be called \emph{dirty} if $\mathcal{D}_{j} \geq \frac{\epsilon}{\beta} \cdot \mathcal{C}_{j}$, and \emph{half-dirty} if $\mathcal{D}_{j} \geq \frac{1}{2} \cdot \frac{\epsilon}{\beta} \cdot \mathcal{C}_{j}$.
\item Following \cite{abboud}, we define a level $i$ to be \emph{critical} if for \emph{every} $j$ that satisfies $0 \leq j \leq i$ we have that ($\mathcal{D}_i^i = \mathcal{D}_{i}$):

\begin{equation} \label{eq1}
\mathcal{D}_{j}^i \geq \frac{\epsilon}{\beta} \cdot \mathcal{C}_{j}^i
\end{equation}

\item Similarly, we define a level $i$ to be \emph{half-critical} if for \emph{every} $j$ that satisfies $0 \leq j \leq i$ we have that:

\begin{equation} \label{eq1half}
\mathcal{D}_{j}^i \geq \frac{1}{2} \cdot \frac{\epsilon}{\beta} \cdot \mathcal{C}_{j}^i
\end{equation}

\noindent Notice that a dirty level is also half-dirty, and a critical level is also half-critical.
\item Define $\mathcal{D} := \sum_j \mathcal{D}_j$ and $\mathcal{C} := \sum_j \mathcal{C}_j$ (where the sum is on all levels). The ``system'' will be called {\em dirty} if $\mathcal{D} \ge \frac{\epsilon}{\beta} \cdot \mathcal{C}$.

\end{itemize}
\end{definition}

\begin{invariant} [the system is not dirty] \label{inv2}
$\mathcal{D} < \frac{\epsilon}{\beta} \cdot \mathcal{C}$
\end{invariant}

\noindent A covering set that is both PC and NC is called \emph{clean}; thus for a clean set $S$,
it holds that (1) $\frac{|Cov(S)|}{c(S)}\geq \beta^{l_{cov}(S)}$, and (2) $\frac{|N_j(S)|}{c(S)} <\beta^{j}$, for every $j > l_{cov}(S)$.
Condition (2) in the particular case $j = l_{cov}(S) + 1$ yields $\frac{|N_{l_{cov}(S) +1}(S)|}{c(S)} < \beta^{l_{cov}(S) +1}$;
thus, since all elements in $Cov(S)$ are covered at level $l_{cov}(S)$, $|Cov(S)| \leq |N_{l_{cov}(S) +1}(S)|$.
We conclude that for a clean covering set $S$:
\begin{equation} \label{cleandom}
\beta^{l_{cov}(S)} \le \frac{|Cov(S)|}{c(S)} < \beta^{l_{cov}(S) +1}.
\end{equation}

\noindent For the next invariant, we introduce the following definition.
\begin{definition} \label{nbr}
For any element $e$, let $\mathcal{F}_e$ denote the collection of sets to which $e$ belongs. Notice that $|\mathcal{F}_e| = f_e$.
\end{definition}

\noindent We will maintain the following invariant for every element,
which says that, for any element $e$, there does not exist any covering set which contains $e$ that covers at a level higher than $l(e)$.
\begin{invariant} \label{inv3}
For any element $e$, $l_{cov}(S) \leq l(e)$, for any $S \in \mathcal{F}_e$.
\end{invariant}

\begin{cl} \label{lem1}
If the system is dirty, a critical and half-critical level exist.
\end{cl}

\begin{proof}
We will prove that if a critical level does not exist, the system is not dirty. If a critical level does not exist, it means that the highest level, which we will denote by $i_0 - 1$, is not critical, thus there exists some $i_1$ such that $\mathcal{D}_{i_1}^{i_0 - 1} < \frac{\epsilon}{\beta} \cdot \mathcal{C}_{i_1}^{i_0 - 1}$. Now, level $i_1-1$ is not critical either, thus there exists some $i_2$ such that $\mathcal{D}_{i_2}^{i_1-1} < \frac{\epsilon}{\beta} \cdot \mathcal{C}_{i_2}^{i_1-1}$. We continue in the same manner until we reach level $i_q = 0$. Summing up each side of the inequalities yields 
$$\mathcal{D} = \mathcal{D}_{0}^{i_0-1} ~=~ \sum_{j=0}^{q-1} \mathcal{D}_{i_{j+1}}^{i_j-1} ~<~ \frac{\epsilon}{\beta} \cdot \sum_{j=0}^{q-1} \mathcal{C}_{i_{j+1}}^{i_j-1} ~=~ \frac{\epsilon}{\beta} \cdot \mathcal{C}_{0}^{i_0 - 1} ~=~ \frac{\epsilon}{\beta} \cdot \mathcal{C},$$

\noindent meaning the system is not dirty. Clearly any critical level is also half-critical, so the claim holds.
\qed
\end{proof}

\begin{definition} {\ } \label{good}
\begin{itemize}
\item Denote by $\mathcal{E}_{i}$ the number of elements that leave the \emph{initial} Cov set of a set at level $i$, meaning $\mathcal{E}_{i} = \mathcal{D}_i \cdot \beta^i$. 
\item For each pair of levels $j$ and $i$ (where $j \leq i$) we define $\mathcal{E}_{j}^{i} := \sum_{q=j}^{i} \mathcal{E}_q$.
\item Denote by $\mathcal{L}_i$ the number of elements that are covered at level $i$, and let $\mathcal{L}_{j}^{i} := \sum_{q=j}^{i} \mathcal{L}_q$ (where $j \leq i$).

\end{itemize}
\end{definition}

\noindent \textbf{The static greedy algorithm:} In our dynamic algorithm, every once in a while we will perform a \emph{partial reset} (sometimes abbreviated to a \emph{reset}), which means running the {\em static greedy} weighted set cover algorithm, on a subset of $\mathcal{U}$ and $\mathcal{F}$ as an input. More specifically, the static greedy algorithm receives as an input a collection of sets $\tilde{\mathcal{F}} \subseteq \mathcal{F}$, and a collection of elements $\tilde{\mathcal{U}} \subseteq \mathcal{U}$.
$\tilde{\mathcal{U}}$ contains elements that we removed from the Cov sets they were contained in, and we want to ``re-cover" them, more ``optimally" in a sense. When the static algorithm has finished its execution, they must be covered by a set in $\tilde{\mathcal{F}}$. In each iteration, a set $S \in \tilde{\mathcal{F}}$ with a minimum ratio of cost over number of elements in $\tilde{\mathcal{U}}$ it contains is found. Then, we move $S$ from $\tilde{\mathcal{F}}$ to the set cover $\mathcal{S}$, and remove all elements that were in $\tilde{\mathcal{U}}$ (and now covered by $S$, to form $Cov(S)$) from $\tilde{\mathcal{U}}$. We do not re-cover elements outside of $\tilde{\mathcal{U}}$ during this process. This is repeated until $\tilde{\mathcal{U}}$ is empty, thus eventually all elements in the original set $\tilde{\mathcal{U}}$ are covered by sets in the original set $\tilde{\mathcal{F}}$, and thus also by the set cover $\mathcal{S}$. In Section \ref{alg} we describe the partial reset procedure, as employed by our algorithm.

\bigskip

\noindent \textbf{Remark.} In DS, the static greedy algorithm will receive two disjoint sets of vertices as input, denoted $W$ and $B$. The set $W$ contains all vertices that are currently not dominated by the dominating set, and when the static algorithm has finished its execution, they must be dominated by a vertex in the set $B \cup W$. In each iteration, a vertex $v \in B \cup W$ with a minimum ratio of cost over number of neighbors in $W$ is found. Then, we move $v$ from $B \cup W$ to the dominating set, and move all vertices that were in $W$ and are currently dominated by $v$ to $B$. This is repeated until $W$ is empty.

\subsection{Algorithm} \label{alg}
Our update algorithm consists of the basic insertion and deletion algorithms (Algorithms \ref{alg1} and \ref{alg2})
and two procedures (Algorithms \ref{alg3} and \ref{alg4}), which we will refer to as the {\em local rise} and {\em partial reset} procedures.
For each set $S$, we shall maintain the $\rho = O(\log_{\beta} n)$ counters $N_0(S), ..., N_\rho (S)$;
see Section \ref{rel} for details regarding why $\rho = O(\log_{\beta} n)$ and not $\Theta(\log_{\beta} n \cdot C)$.
For efficiency purposes, we shall actually maintain only ($1+\epsilon)$-approximate counters (see Section \ref{updatetime}).
Every time $|Cov(S)|$ decrements (which, as shown below, could happen due to element deletion or local rise), we update $\mathcal{D}_{l_{cov}(S)}$ accordingly. Every time a non-covering set joins the SC (which could happen due to element insertion, local rise, or partial reset) or a covering set leaves the SC (which could happen only due to a partial reset) we update $\mathcal{S}$ and $\mathcal{C}$ accordingly. 
As will become evident later, the maintenance of $\mathcal{D}$, $\mathcal{S}$, $\mathcal{C}$ and all the $\mathcal{D}_j$, $\mathcal{S}_j$, $\mathcal{C}_j$, $\mathcal{E}_j$ and $\mathcal{L}_j$ counters can be done efficiently in a straightforward way.

\paragraph{Insertion and deletion algorithms.}
Upon insertion of element $e$, we add $e$ to $Cov(S)$ where $S = \arg \max_{S'} \{l_{cov}(S') : S' \in \mathcal{F}_e\}$ (breaking ties arbitrarily), in order to maintain Invariant \ref{inv3}. If none of the sets in $\mathcal{F}_e$ are in the set cover, then we choose one with minimum weight, $S$, add it to the set over at level $\lfloor \log_{\beta}(\frac{1}{c(S)}) \rfloor$ (so that Equation (\ref{cleandom}) would hold), and add $e$ to $Cov(S)$. We also update $N_j(S)$ for all $S \in \mathcal{F}_e$ and for all relevant $j$'s. Upon deletion of element $e$, we simply update $Cov(S)$ where $e \in Cov(S)$, and update $N_j(S)$ for all $S \in \mathcal{F}_e$ and for all relevant $j$'s. For DS, upon insertion of edge $(u,v)$, we first update $N_j(u)$ and $N_j(v)$ for all relevant $j$'s. {If there is a violation of Invariant \ref{inv3}, it is   because $l_{dom}(u) > l(v)$ or $l_{dom}(v) > l(u)$ hold ($l_{dom}(v)$ is the analog of $l_{cov}(S)$). In the former (resp., latter) case we remove $v$ (resp., $u$) from the Dom set it belongs to and move it to $Dom(u)$ (resp., $Dom(v)$)}, the analog of $Cov(S)$. As a result we must update all $N_j(w)$'s for each $w$ in $N(v)$ (resp., $N(u)$), and for all relevant $j$'s. Upon deletion of edge $(u,v)$, if neither dominates the other, then we simply update $N_j(u)$ and $N_j(v)$ for all relevant $j$'s. If (without loss of generality) $u$ dominated $v$, then $v$ joins the Dom of its highest-level dominating neighbor; this guarantees that Invariant \ref{inv3} is maintained. As a result we must update all $N_j(w)$'s for each $w \in N(v)$, and for all relevant $j$'s.  
See Algorithms \ref{alg1} and \ref{alg2} for the pseudo-code of the element insertion and deletion algorithms;
the full description of these algorithms relies on the local rise and partial reset procedures,  described next.

\paragraph{Local rise.} Once a set $S$ becomes PD, and in particular $j$-PD for some $j$ (see Definition \ref{positive}), we create $Cov(S)$ at level $j+1$ with $N_{j+1}(S)$ as $Cov(S)$. If $S$ is $j$-PD with respect to multiple levels $j$, we take the highest level $j$. Generally, 
if after any update step there are multiple PD sets (not necessarily with respect to the same level $j$), we begin by creating the pair with the highest level $j$. 
This new creation of $Cov(S)$ is called a \emph{local rise} of $S$ to level $j+1$. For DS, no significant changes are necessary. Once a vertex $v$ is $j$-PD, we create $Dom(v)$ at level $j+1$, with $N_{j+1}(v)$ as $Dom(v)$. See Algorithm \ref{alg3} for the pseudo-code of the local rise procedure.

\begin{obs} \label{obclean}
Following a local rise of set $S$, it is NC and not PD, and Equation (\ref{cleandom}) holds.  
\end{obs}

\begin{proof}
{Suppose that $S$ performs a local rise to level $j+1$. Note that since $S$ was $j$-PD prior to the local rise, $\frac{|N_j(S)|}{c(S)}$ was $\geq \beta^{j+1}$. Since $N_{j+1}(S)$ is taken as the new $Cov(S)$, and $N_j(S) \subseteq N_{j+1}(S)$, we get that indeed $\frac{|Cov(S)|}{c(S)} \geq \beta^{j+1}$ following the local rise. On the other hand $\frac{|N_l(S)|}{c(S)} < \beta^{l+1}$ for every $l > j$, 
otherwise
$S$ would also be $l$-PD prior to the local rise, contradicting the fact that the rise of $S$ was to the highest possible level. In particular, $\frac{|N_{j+1}(S)|}{c(S)}$ was $<\beta^{j+2}$, and again we remind that $N_{j+1}(S)$ is taken as the new $Cov(S)$. Thus, we get that following the local rise to $j+1$, $\frac{|Cov(S)|}{c(S)} \geq \beta^{j+1}$, and $\frac{|Cov(S)|}{c(S)} < \beta^{j+2}$.}
\qed
\end{proof}

\begin{obs} \label{obs1}
For any covering set $S$, $\frac{|Cov(S)|}{c(S)} < \beta^{l_{cov}(S) + 2}$.
\end{obs}

\begin{proof}
{If $\frac{|Cov(S)|}{c(S)} \geq \beta^{l_{cov}(S) + 2}$, then since $Cov(S) \subseteq N_{l_{cov}(S)+1}(S)$, it is $(l_{cov}(S)+1)$-PD, and thus it would perform a local rise to level $\geq l_{cov}(S)+2$.}
\qed
\end{proof}

\noindent We will argue later (see Claim \ref{clm0}) that a local rise cannot trigger any violation to Invariant \ref{inv3}.

\subsubsection{Partial Reset}

\noindent When the system is dirty, a half critical level exists by Claim \ref{lem1}, and we shall perform a \emph{partial reset} up to one of the existing half-critical levels, denoted by $i_{crit}$. 
Determining if the system is dirty is straightforward, given the values of $\mathcal{D}$ and $\mathcal{C}$, which can be efficiently maintained in the obvious way. 
In Section \ref{updatetime} we describe how to compute a half-critical level, as well as the entire partial reset procedure, efficiently. 

At the outset, we perform an ``imaginary" reset on the entire initial system (with no elements), where we may set $i_{crit}$ to be the highest level, and initially $Cov(S) = \emptyset$ and for all sets $S$. A reset starts by emptying $Cov(S)$ for any set $S$ such that $l_{cov}(S) \leq i_{crit}$. We transfer all the elements that were taken out of these sets to the set $\tilde{\mathcal{U}}$, and all of the sets $S$ that contain at least one element from $\tilde{\mathcal{U}}$ which satisfy $l_{cov}(S) \leq i_{crit}$ (the emptied out sets plus perhaps some that were not in the set cover) to the set $\tilde{\mathcal{F}}$. \textbf{We reset all of the $\mathcal{D}_i$ and $\mathcal{E}_i$ counters for all $i \leq i_{crit}$}.

\bigskip

\noindent \textbf{Remark.} In DS, a reset starts by emptying $Dom(v)$ for any vertex $v$ such that $l_{dom}(v) \leq i_{crit}$. We transfer every vertex $u$ such that $l(u) \leq i_{crit}$ to $W$, and transfer to the set $B$ all neighbors of vertices in $W$ that dominate at a level up to $i_{crit}$ (including those that do not dominate at all).  

\bigskip

\noindent We apply the static algorithm on $\tilde{\mathcal{U}}$ and $\tilde{\mathcal{F}}$ with the following clarification:

\paragraph{Placement into levels.}
Consider a set $S$ that becomes covering as a result of the static greedy algorithm. $S$ is moved to $\mathcal{S}$, and all elements which are contained in $S$ and were in $\tilde{\mathcal{U}}$ now form $Cov(S)$, which is placed at level $\lfloor \log_{\beta}(\frac{|Cov(S)|}{c(S)}) \rfloor$.

\paragraph{Basic properties of partial reset.} We will argue later (see Claim \ref{clm0}) that a partial reset cannot trigger any violations to Invariant \ref{inv3}. 
We record the following readily verified observation for further usage.

\begin{obs} \label{obs2}
Only elements in $\tilde{\mathcal{U}}$ can change their level due to a partial reset. 
\end{obs}

{
\noindent{\bf Remark.}
We shall view each element in $\tilde{\mathcal{U}}$ as changing its level following a partial reset, even if it happens to remain covered in the same level. This tweak may only increase the number of element level changes; thus by upper bounding the number of element level changes under this tweak, we will bound also the actual number of element level changes.
}

\bigskip

\noindent We make the following two related observations.
\begin{obs} \label{parclean}
Any set $S$ that is covering due to a partial reset is NC and is not PD, and Equation (\ref{cleandom}) holds.  
\end{obs}
\begin{proof}
Let $l_{cov}(S) = l$ be the level of set $S$ following the reset.
We first note that $\beta^{l_{cov}(S)} \leq \frac{|Cov(S)|}{c(S)} \leq \beta^{l_{cov}(S)+1}$ (see \emph{Placement into levels} paragraph). 
To complete the proof, we next argue that $\frac{|N_{j}(S)|}{c(S)} < \beta^{j+1}$ for any $j > l_{cov}(S)$.
Indeed, suppose for contradiction that $\frac{|N_{j}(S)|}{c(S)} \geq \beta^{j+1}$ for some $j > l_{cov}(S)$. 
If $j \le i_{crit} + 1$, that would mean that $S$ has at least $\beta^{j+1} \cdot c(S)$ elements at level $<j$ that it can cover, thus it would be placed by the partial reset at a level at least $j+1$, which is larger than $l_{cov}(S)$, a contradiction.
If on the other hand $j > i_{crit} + 1$, then by Observation \ref{obs2} all elements in $N_{j}(S)$ after the reset were in $N_j(S)$ also prior to it (some might have left it), but that would mean that $S$ was $j$-PD prior to the reset, a contradiction to Invariant \ref{inv1}. 
\qed
\end{proof}

\begin{obs} \label{ob:added}
A set that was in $\tilde{\mathcal{F}}$ at the beginning of the partial reset cannot become covering at a level $l > i_{crit} + 1$ following the reset.
\end{obs}
\begin{proof}
If any set $S$ that was in $\tilde{\mathcal{F}}$ at the beginning of the partial reset becomes covering at a level $l > i_{crit} + 1$ following the reset, then by Observation \ref{obs2} that would mean that $S$ was $(l-1)$-PD prior to the reset.
\qed
\end{proof}

\begin{cl} \label{clm0}
Only a single element level change occurs following a single element update {\em due to Invariant \ref{inv3}}, and it occurs only following an element insertion.
\end{cl}

\begin{proof}
Following an element insertion of $e$, we scan the up-to-$f$ sets which contain $e$, and $e$ joins the covering set at the highest level, to maintain Invariant \ref{inv3}. We consider this as an element level change for $e$. Now, assume by contradiction that a partial reset triggers a violation to Invariant \ref{inv3}. Thus, element $e$ is in $Cov(S_1)$ and not $Cov(S_2)$ after the reset, even though $S_1,S_2 \in \mathcal{F}_e$ and $l_{cov}(S_1) < l_{cov}(S_2)$. We know that $S_1 \in \tilde{\mathcal{F}}$ for this reset. If $S_2 \in \tilde{\mathcal{F}}$ as well, then by description of the algorithm $e$ will join $Cov(S_2)$. If $S_2 \notin \tilde{\mathcal{F}}$, then $l_{cov}(S_2) > i_{crit}$. But since $l(e) \leq i_{crit}$ (it was in $\tilde{\mathcal{U}}$), that means there was a violation to Invariant \ref{inv3} before the reset. Next, it is clear that a local rise cannot trigger a violation to Invariant \ref{inv3}, since when $S$ performs a local rise to level $j$, it takes to $Cov(S)$ all elements it contains which are covered at a level less than $j$. Finally, it is trivial that no such violation can occur due to an element deletion. We conclude that only following an insertion an element changes its level in order to maintain Invariant \ref{inv3} (the inserted element), and no violations to this invariant occur due to resets, local rises or deletions.
\qed
\end{proof}

\noindent \textbf{Remark.} In DS, only a single dominated level change (a vertex changing its dominated level) may occur following a single edge update due to Invariant \ref{inv3}, and it may occur only following an edge insertion. 

\subsubsection{Finding the Relevant Levels} \label{rel}

\noindent In the algorithm we occasionally update $N_j(S)$ for all relevant $j$'s. Naively, the number of levels is $O(\log_{\beta}(n \cdot C))$, however we want to show that for each set $S$, there are only $O(\log_{\beta}n)$ (to be precise $\log_{\beta}n + 3$) levels $j$ which we need to update the $N_j(S)$ counter. For each set $S$, we define a minimum relevant level and maximum relevant level, denoted $j_{min}^S$ and $j_{max}^S$ respectively. Let $j_{min}^S = \lfloor \log_{\beta}(\frac{1}{c(S)}) \rfloor - 1$ and $j_{max}^S = \lceil \log_{\beta}(\frac{n}{c(S)}) \rceil - 1$. We first claim that $S$ cannot be $j$-PD for any $j > j_{max}^S$. Indeed, if it were, then by definition $|N_{j}(S)| \geq \beta^{j+1} \cdot c(S)$. But $\beta^{j+1} \geq \beta^{\log_{\beta}(\frac{n}{c(S)})+1} = \beta \cdot \frac{n}{c(S)}$, meaning that $|N_{j}(S)| \geq \beta \cdot n$, and the number of elements is upper bounded by $n$, so this is a contradiction. Next, we claim that if $S$ is $j$-PD for any $j < j_{min}^S$, it is also $j_{min}^S$-PD, and since we perform local rises from high to low levels, $S$ being $j$-PD is irrelevant (a local rise to $\geq j_{min}^S+1$ would occur). Indeed, if $S$ is $j$-PD, then clearly $|N_j(S)| \geq 1$. But $S$ would be $j_{min}^S$-PD as well, since $|N_{j_{min}^S}(S)| \geq 1$, and $\beta^{j_{min}^S+1} \cdot c(S) \leq \beta^{\log_{\beta}(\frac{1}{c(S)})} \cdot c(S) = 1$, thus $|N_{j_{min}^S}(S)| \geq \beta^{j_{min}^S+1} \cdot c(S)$ and $S$ is $j_{min}^S$-PD. Since $j_{max}^S - j_{min}^S +1 \leq \log_{\beta}n + 3$, we get that for each set $S$ there are only $\log_{\beta}n + 3$ relevant levels.

\begin{algorithm} 
\caption{Insert $(e,\mathcal{F}_e)$} \label{alg1}
\begin{algorithmic} 
\STATE {$S = \arg \max_{S'} \{l_{cov}(S') : S' \in \mathcal{F}_e\}$}
\STATE {$Cov(S) \leftarrow Cov(S) \cup \{e\}$}
\STATE $l(e) \leftarrow l_{cov}(S)$
\STATE ChangeLevel$(e)$
\STATE Check if a partial reset is due, and if so perform one
\end{algorithmic}
\end{algorithm}

\begin{algorithm}
\caption{Delete $(e,\mathcal{F}_e)$} \label{alg2}
\begin{algorithmic}
\STATE {Let $S$ be the set covering $e$.}
\STATE {$Cov(S) \leftarrow Cov(S) \setminus \{e\}$}
\STATE ChangeLevel$(e)$
\STATE Check if a partial reset is due, and if so perform one
\end{algorithmic}
\end{algorithm}

\begin{algorithm}
\caption{LocalRise $(S,j,N_{j+1}(S))$} \label{alg3}
\begin{algorithmic}
\STATE $Cov(S) \leftarrow N_{j+1}(S)$
\STATE Create $Cov(S)$ at level $j+1$
\FOR {$e \in Cov(S)$}
\STATE ChangeLevel$(e)$
\ENDFOR
\STATE  Check if a partial reset is due, and if so perform one
\end{algorithmic}
\end{algorithm}

\begin{algorithm}
\caption{PartialReset $(i_{crit})$} \label{alg4}
\begin{algorithmic}
\STATE $\tilde{\mathcal{U}} \leftarrow \{e:l(e) \leq i_{crit}\}$
\STATE $\tilde{\mathcal{F}} \leftarrow \{S:l_{cov}(S) \leq i_{crit}\}$
\FOR {$S \in \tilde{\mathcal{F}}$}
\STATE $Cov(S) \leftarrow \emptyset$
\STATE $\mathcal{S} \leftarrow \mathcal{S} \setminus \{S\}$
\ENDFOR
\WHILE {$\tilde{\mathcal{U}} \neq \emptyset$}
\STATE Let $S$ be a set in $\tilde{\mathcal{F}}$ such that $\frac{c(S)}{|S \cap \tilde{\mathcal{U}}|}$ is minimized.
\STATE $Cov(S) \leftarrow S \cap \tilde{\mathcal{U}}$ 
\STATE Create $Cov(S)$ at level $\lfloor \log_{\beta}(\frac{|Cov(S)|}{c(S)}) \rfloor$
\STATE $\mathcal{S} \leftarrow \mathcal{S} \cup \{S\}$
\STATE $\tilde{\mathcal{U}} \leftarrow \tilde{\mathcal{U}} \setminus Cov(S)$
\ENDWHILE
\end{algorithmic}
\end{algorithm}

\begin{algorithm} 
\caption{ChangeLevel $(e)$} \label{algcl}
\begin{algorithmic} 
\STATE update relevant counters ($\mathcal{D},\mathcal{L},\mathcal{E}$)
\STATE pd $\leftarrow$ FALSE
\FOR {$j=$ highest relevant level of a set in $\mathcal{F}_e$ \textbf{down to} lowest relevant level of a set in $\mathcal{F}_e$}
\FOR {$S \in \mathcal{F}_e$ and $j$ is a relevant level of $S$}
\STATE $update$ $N_j(S)$
\IF {$S$ is $j$-PD \textbf{and} pd = FALSE}
\STATE $S' \leftarrow S$; $j' \leftarrow j$; pd $\leftarrow$ TRUE 
\ENDIF
\ENDFOR
\ENDFOR
\IF {pd = TRUE}
\STATE LocalRise $(S',j',N_{j'+1}(S'))$
\ENDIF
\end{algorithmic}
\end{algorithm}

\newpage

\noindent For DS, the local rise, partial reset and level change algorithms are completely analogous to the ones for SC. Since the insert and delete algorithms have some differences from their respective SC algorithms, we present the insert and delete algorithms for DS.

\begin{algorithm} 
\caption{Insert $(u,v)$} \label{alg1ds}
\begin{algorithmic} 
\IF {$l_{dom}(u) > l(v)$ \textbf{or} $l_{dom}(v) > l(u)$ (assume WLOG $l_{dom}(u) > l(v)$, cannot be both)}
\STATE {$Dom(z_v) \leftarrow Dom(z_v) \setminus \{v\}$ (assume $v$ belonged to $Dom(z_v)$)}
\STATE {$Dom(u) \leftarrow Dom(u) \cup \{v\}$}
\STATE ChangeLevel$(v)$
\ELSE 
\STATE pd = FALSE
\FOR {$j=\max\{j_{max}^v,j_{max}^u\}$ \textbf{down to} $\min\{j_{min}^v,j_{min}^u\}$}
\FOR {$w \in \{v,u\}$}
\IF {$j$ is a relevant level of $w$}
\STATE update $N_j(w)$
\IF {$w$ is $j$-PD \textbf{and} pd = FALSE}
\STATE $w' \leftarrow w$; $j' \leftarrow j$; pd $\leftarrow$ TRUE 
\ENDIF
\ENDIF
\ENDFOR
\ENDFOR
\IF {pd = TRUE}
\STATE LocalRise $(w',j',N_{j'+1}(w'))$
\ENDIF
\ENDIF
\STATE Check if a partial reset is due, and if so perform one
\end{algorithmic}
\end{algorithm}

\begin{algorithm}
\caption{Delete $(u,v)$} \label{alg2ds}
\begin{algorithmic}
\IF {$v \in Dom(u)$ or $u \in Dom(v)$ (assume WLOG $v \in Dom(u)$)}
\STATE Scan $N(v)$ to find $z$, a highest level dominator (if none exists, $z \leftarrow v$)
\STATE {$Dom(z) \leftarrow Dom(z) \cup \{v\}$} (if $z = v$, $Dom(z) = \{v\}$)
\FOR {$j=j_{max}^u$ \textbf{down to} $l_{dom}(u)+1$}
\STATE update $N_j(u)$
\ENDFOR
\STATE ChangeLevel$(v)$
\ELSE
\FOR {$j=\max\{j_{max}^v,j_{max}^u\}$ \textbf{down to} $\min\{j_{min}^v,j_{min}^u\}$}
\FOR {$w \in \{v,u\}$}
\IF {$j$ is a relevant level of $w$}
\STATE update $N_j(w)$
\ENDIF
\ENDFOR
\ENDFOR
\ENDIF
\STATE Check if a partial reset is due, and if so perform one
\end{algorithmic}
\end{algorithm}

\noindent 

\newpage

%% file: sec456.tex
\section{Update Time Analysis} \label{time}

\noindent In Section \ref{amorsubsub} we provide an upper bound of $O(\epsilon^{-3} \cdot \ln n)$ to the amortized number of {\em element level changes} in the update sequence, meaning the number of times the procedure \emph{ChangeLevel} has been called (see Algorithms \ref{alg1} through \ref{alg2ds}). We will then bound the number of set level changes as a direct corollary. In Section \ref{omega}, we prove that the upper bound for element level changes is tight. Finally, in Section \ref{updatetime}, we show that we can handle each element level change in $O(\epsilon^{-2} \cdot f)$ amortized time, and that the total amortized update time due to element level changes of $O(\epsilon^{-5} \cdot \ln n \cdot f)$ is the bottleneck of the algorithm.

\subsection{$O_{\epsilon}(\ln n)$ Amortized Element Level Changes} \label{amorsubsub}

\noindent We define a partial reset that some set $S$ \emph{participates} in as a reset where in the beginning of it $S \in \tilde{\mathcal{F}}$. We begin by proving the following two claims:

\begin{cl} \label{clm1}
Following a reset up to level $i_{crit}$ in which $S$ participates in, $\frac{|N_l(S)|}{c(S)}<\beta^l$ for any $l \leq i_{crit}$.
\end{cl}

\begin{proof}
Assume by contradiction that $|N_l(S)| \geq c(S) \cdot \beta^l$ for some $l \leq i_{crit}$. Since $l \leq i_{crit}$, all elements in $N_l(S)$ right after the reset were part of the reset in the set $\tilde{\mathcal{U}}$, by Observation \ref{obs2}. Since $S$ participated in the reset, $S$ was in the set $\tilde{\mathcal{F}}$. Thus, $S$ would have taken all elements in $N_l(S)$ to $Cov(S)$ at level at least $l$ during the partial reset, therefore all of them would be covered at level $\geq l$ and we would have $|N_l(S)|=0$ following the reset, a contradiction.
\qed
\end{proof}

\begin{cl} \label{clm2}
If the level of an element $e$ drops due to a partial reset up to level $i_{crit}$, then $i_{crit} \geq l(e)$ (where $l(e)$ is the level before the drop). Thus $|N_l(S)|$ cannot grow as a result of a partial reset up to level $i_{crit} \leq l-1$.
\end{cl}

\begin{proof}
Consider a partial reset up to level $i_{crit}$, and let $e$ be any element such that $l(e) > i_{crit}$. Then $e$ will not be in $\tilde{\mathcal{U}}$ in such a reset, thus by Observation \ref{obs2} it is not possible for $e$ to change its level due to this reset. For $|N_l(S)|$ to grow as a result of a reset with $i_{crit} \leq l-1$, an element $e$ which can be covered by $S$ that satisfies $l(e) \geq l > i_{crit}$ must reduce its level below $l$, but we have just proved that this is impossible.
\qed
\end{proof}

\noindent We make the following simple observation, which relies on Claim \ref{clm2}.

\begin{obs} \label{obs3}

The level of an element $e$ can drop only due to a partial reset up to level $i_{crit} \geq l(e)$.

\end{obs}

\begin{proof}
The level of an element $e$ cannot drop due to an insertion, deletion or local rise. By Claim \ref{clm2}, the only way for it to drop due to a reset is if $i_{crit} \geq l(e)$.
\qed
\end{proof}

\noindent The following observation relies on Observation \ref{obs2}.

\begin{obs} \label{nochange}
$|N_{l}(S)|$ cannot change as a result of a partial reset in which $S$ did not participate in, for any $l$.
\end{obs}

\begin{proof}
If $S$ did not participate in a reset, it means that $S \notin \tilde{\mathcal{F}}$, implying that $S$ does not contain any element in $\tilde{\mathcal{U}}$. By Observation \ref{obs2}, only elements in $\tilde{\mathcal{U}}$ can change their level due to a reset. Therefore, none of the elements contained in $S$ change their level due to the reset, so indeed $|N_l(S)|$ does not change for any $l$.
\qed
\end{proof}

\begin{lem} \label{first}
The amortized number of element level changes due to local rises to level $j$ is $O(\epsilon^{-1})$, for any $j$. 
\end{lem}

\begin{proof}
Let $l = j-1$. We use a token scheme, where each element insertion gives $c_{\epsilon}$ tokens ($c_{\epsilon}$ will be taken as $\frac{\beta^2}{\epsilon}$ later).  
Once an element changes the level in which it is covered to level $j = l+1$ due to a local rise, it pays one token. Thus, we will prove the lemma by showing that the system does not run out of tokens.

Consider some arbitrary local rise of set $S$ to level $l+1$. Denote the time step of this local rise by $t'$. At time step $t'$, we have $\frac{|N_l(S)|}{c(S)} \geq \beta^{l+1}$. Denote by $t$ the time step of the last time before $t'$ in which a reset up to level $i_{crit} \geq l$ in which $S$ participated in, has occurred; recall that there is an imaginary reset on the entire system at time $0$, and so $t$ is well-defined. By Claim \ref{clm1}, we conclude that between time steps $t$ and $t'$, $S$ has obtained more than $c(S) \cdot (\beta^{l+1} - \beta^{l}) = c(S) \cdot \epsilon \cdot \beta^{l}$ new low-level (lower than $l$) elements which it can cover. By Observation \ref{obclean}, following the local rise $|Cov(S)| <c(S) \cdot \beta^{l+2}$. Thus, by setting $c_{\epsilon}=\frac{\beta^2}{\epsilon}$, if we could show that each of the new low-level elements in $S$ still has its $c_{\epsilon}$ tokens, we could redistribute them such that each element in the new $Cov(S)$ would have a token for the local rise.

Consider such an arbitrary low-level element denoted by $e$. The only way for $e$ to join $N_l(S)$ (thereby incrementing $|N_l(S)|$) is either by insertion of element $e$ or a partial reset. However, since all partial resets in the time interval $(t,t')$ are either up to level $\leq l-1$ or $S$ did not participate in them (by definition of $t$ and $t'$), Claim \ref{clm2} and Observation \ref{nochange} respectively imply that $e$ cannot be a new low-level element contained in $S$ due to a partial reset. We conclude that each of the new low-level elements are a result of an element insertion. Notice that each insertion can raise $|N_{l}(S)|$ by only one, and each new low-level element due to such an insertion arrives with a fresh set of tokens. Now, we claim that at time $t'$, $e$ still has its $c_{\epsilon}$ tokens (and did not spend any on some prior local rise). Recall that it received these tokens in an insertion in the {\em time interval $(t,t')$}, i.e., sometime after $t$ and before $t'$. Had $e$ not had $c_{\epsilon}$ tokens at time $t'$, it would mean that sometime in the time interval $(t,t')$, $e$ was covered at level $l+1$. For $e$ to be part of a local rise at $t'$ to level $l+1$ though, the level of $e$ must have dropped sometime before $t'$, in a reset up to level $\geq l+1$ by Observation \ref{obs3}. This reset would have to be up to level at least $l+1$, and since $e$ would be in $\tilde{\mathcal{U}}$ and it is contained in $S$ then $S$ would have to participate in this reset, a contradiction to the definition of $t$. Thus $e$ still has the $c_{\epsilon}$ tokens that it received.

Thus, for a local rise of $S$ to level $l+1$, we have overall acquired at least $c(S) \cdot c_{\epsilon} \cdot (\epsilon \cdot \beta^{l})$ tokens. By Observation \ref{obclean}, $|Cov(S)| <c(S) \cdot \beta^{l+2}$. By setting $c_{\epsilon}=\frac{\beta^2}{\epsilon}$, we would have enough tokens for all of $Cov(S)$. We conclude that if each element insertion gives $\frac{\beta^2}{\epsilon}$ tokens, we can pay one token for each element level change due to local rises to level $j = l+1$ in the sequence. 
\qed
\end{proof}

\noindent \textbf{Remark.} In DS, each edge update will give $c_{\epsilon}$ tokens to each of the two endpoints of the deleted/inserted edge. Here too, between time steps $t$ and $t'$ the vertex $v$ (which is performing the local rise) has obtained more than $c(v) \cdot \epsilon \cdot \beta^l$ new low-level neighbors. Denote one of them by $u$. As in SC, $u$ cannot join $N_l(v)$ due to a partial reset. However, it can join due to an edge insertion \emph{or deletion}. Insertion is trivial, if $u$ and $v$ were not neighbors and then edge $(u,v)$ was inserted. For deletion, if $u$ was dominated by some $x$ at a level at least $l$, and then edge $(u,x)$ was deleted, then $u$ will be dominated by its next highest dominating neighbor, which could be lower than $l$, and as a result it would join $N_l(v)$. Since we can pay tokens for each insertion and deletion, the lemma holds.

\begin{lem} \label{lem2}
Denote by $K$ the total number of insertions and deletions in the sequence. Then there are $\leq [\frac{\beta^2}{\epsilon} \cdot O(\log_{\beta}n)] \cdot K$ element level changes due to local rises in the sequence.
\end{lem}

\begin{proof}
Using a similar token scheme, each element insertion gives $\frac{\beta^2}{\epsilon} \cdot O(\log_{\beta} n)$ tokens. We claim that there are only $O(\log_{\beta} n)$ relevant levels for $e$, where we define a relevant level for $e$ as a level that $e$ can be part of a local rise to. Denote by $S_1$ the set in $\mathcal{F}_e$ with minimum cost. The highest level which $e$ can climb to via local rise is clearly $\lfloor \log_{\beta} (\frac{n}{c(S_1)}) \rfloor$ (no more than $n$ elements). The lowest level which $e$ can climb to via local rise is $\lfloor \log_{\beta} (\frac{1}{c(S_1)}) \rfloor - 1$. Indeed, if we assume that a local rise to a lower level, $j$, in which $e$ is part of, can occur, then clearly $|N_{\lfloor \log_{\beta} (\frac{1}{c(S_1)}) \rfloor - 2}(S_1)| \geq 1$. But $\beta^{\lfloor \log_{\beta} (\frac{1}{c(S_1)}) \rfloor - 1} \cdot c(S_1) \leq \beta^{-1} < 1$, so we get that $|N_{\lfloor \log_{\beta} (\frac{1}{c(S_1)}) \rfloor - 2}(S_1)| > \beta^{\lfloor \log_{\beta} (\frac{1}{c(S_1)}) \rfloor - 1} \cdot c(S_1)$, so $S_1$ is $(\lfloor \log_{\beta} (\frac{1}{c(S_1)}) \rfloor - 2)$-PD, and since we perform local rises from high to low levels, we will have a local rise to $\lfloor \log_{\beta} (\frac{1}{c(S_1)}) \rfloor - 1$ and not to $j$. Thus, there are $\lfloor \log_{\beta} (\frac{n}{c(S_1)}) \rfloor - \lfloor \log_{\beta} (\frac{1}{c(S_1)}) \rfloor + 2 = O(\log_{\beta}n)$ relevant levels for $e$. 

Therefore, we can target $\frac{\beta^2}{\epsilon}$ separate tokens for each of the $O(\log_{\beta} n)$ relevant levels for $e$. When an element $e$ is inserted, then the system receives $\frac{\beta^2}{\epsilon}$ tokens labelled $j$ for each $j$ that is a relevant level for $e$. By Lemma \ref{first}, once an element changes the level in which it is covered due to a local rise, it can pay one token (with the corresponding label), and the lemma holds.
\qed
\end{proof}

\noindent \textbf{Remark.} In DS, we want to give each of the two endpoints of the inserted/deleted edge, denoted by $v$ and $u$, $\frac{\beta^2}{\epsilon} \cdot O(\log_{\beta} n)$ tokens. Consider the vertex $v$, and denote by $v_1$ its neighbor with minimum cost. For the same reason as in SC, the highest level which $v$ can climb to via local rise is clearly $\lfloor \log_{\beta} (\frac{n}{c(v_1)}) \rfloor$ ($n$ vertices). The lowest level which $v$ can climb to via local rise is $\lfloor \log_{\beta} (\frac{1}{c(v_1)}) \rfloor - 1$. However, the neighbors of $v$ dynamically change, as opposed to SC where the same sets can cover element $e$ while it is activated. Notice though that the costs are static, thus if $v_1$ is no longer the minimum cost neighbor of $v$, and now $v_1'$ is, then we either deleted edge $(v_1,v)$ or inserted edge $(v_1',v)$. Either way $v$ is an endpoint of an updated edge and as a result receives another batch of $\frac{\beta^2}{\epsilon} \cdot O(\log_{\beta} n)$ tokens, which is now enough to climb from $\lfloor \log_{\beta} (\frac{1}{c(v_1')}) \rfloor - 1$ to $\lfloor \log_{\beta} (\frac{n}{c(v_1')}) \rfloor$.

\bigskip

\noindent The next lemma shows that we have tokens also for element level changes due to partial resets. 
To be more specific, it shows that we have a fresh token for each of the elements entering the reset in the set $\tilde{\mathcal{U}}$; although some of them may not change their covered level as a result of a partial reset, we count every such element in $\tilde{\mathcal{U}}$ as changing its covered level by the partial reset (see the remark following Observation \ref{obs2}). We will raise the number of tokens we hand out per each element update (with respect to the one used in the proof of Lemma \ref{lem2}) to achieve that.

\begin{lem} \label{lem3}
Denote by $K$ the total number of insertions and deletions in the sequence. Then there are $\leq [\frac{\beta^2}{\epsilon}(1+\frac{2\beta^3}{\epsilon})\cdot O(\log_{\beta}n) + \frac{2\beta^3}{\epsilon}] \cdot K$ element level changes due to local rises and partial resets in the sequence.
\end{lem}
\begin{proof}
We use a token scheme similar to that used in the proof of Lemma \ref{lem2}. Each element insertion gives $\frac{\beta^2}{\epsilon}(1+\frac{2\beta^3}{\epsilon})\cdot O(\log_{\beta}n)$ tokens, and each element deletion gives $\frac{2\beta^3}{\epsilon}$ tokens. Once an element changes the level in which it is covered due to a local rise or a partial reset, it pays one token. Thus, we will prove the lemma by showing that the system does not run out of tokens.

By Lemma \ref{lem2}, $\frac{\beta^2}{\epsilon} \cdot O(\log_{\beta}n)$ tokens per insertion would be enough to deal with element level changes due to local rises, meaning that every element that changes its covered level due to a local rise can afford to pay one token. By multiplying this by $(1+\frac{2\beta^3}{\epsilon})$, we have that each element that changes its covered level due to a local rise can pay $(1+\frac{2\beta^3}{\epsilon})$ tokens. Each element $e$ which is part of a local rise to $Cov(S)$ will spend one token to change levels, and will leave behind the additional $\frac{2\beta^3}{\epsilon}$ tokens for the level of its previous Cov set, denoted $l_{cov}(S')$. Upon deletion of element $e$, where say $e$ was in $Cov(S')$, then $e$ gives the $\frac{2\beta^3}{\epsilon}$ tokens to $l_{cov}(S')$. Notice that due to an insertion, an element cannot leave its covering set. We conclude that each element $e$ that leaves its Cov set $S'$ leaves $\frac{2\beta^3}{\epsilon}$ tokens for its previous level, $l_{cov}(S')$, and our goal is to show that given a reset up to level $i_{crit}$, the levels up to $i_{crit}$ have obtained altogether at least $\mathcal{L}_0^{i_{crit}}$ tokens (see definition \ref{good}), so that we can redistribute the tokens such that each element up to level $i_{crit}$ receives one token. 

Recall that once an element $e$ leaves its covering set $S'$ at a level $i$, we increment $\mathcal{D}_i$ by $\frac{1}{\beta^i}$ and $\mathcal{E}_i$ by $1$. Since $i_{crit}$ is half-critical, for every $j \leq i_{crit}$ we have $\mathcal{D}_j^{i_{crit}} \geq \frac{1}{2} \cdot \frac{\epsilon}{\beta} \cdot \mathcal{C}_j^{i_{crit}}$. We argue inductively on $j$, for any $j = i_{crit}, i_{crit} - 1, \ldots, 0$ that we can pay one token for each element covered at a level between $j$ and $i_{crit}$, by using \emph{only} tokens obtained from elements that left levels between $j$ and $i_{crit}$ that were necessary to raise $\mathcal{D}_{j}^{i_{crit}}$ to \emph{exactly} $\frac{1}{2} \cdot \frac{\epsilon}{\beta} \cdot \mathcal{C}_{j}^{i_{crit}}$. 

For the basis $j=i_{crit}$, since level $i_{crit}$ is half-critical (and thus half-dirty as well), $\mathcal{D}_{i_{crit}} \geq \frac{1}{2} \cdot \frac{\epsilon}{\beta} \cdot \mathcal{C}_{i_{crit}}$.
Since each leaving element gives $\frac{2\beta^3}{\epsilon}$ tokens, level $i_{crit}$ has obtained $\mathcal{D}_{i_{crit}} \cdot \beta^{i_{crit}} \cdot \frac{2\beta^3}{\epsilon}$ tokens. And since level $i_{crit}$ is half-dirty, it has obtained $\geq \frac{1}{2} \cdot \frac{\epsilon}{\beta} \cdot \mathcal{C}_{i_{crit}} \cdot \beta^{i_{crit}} \cdot \frac{2\beta^3}{\epsilon} = \mathcal{C}_{i_{crit}} \cdot \beta^{i_{crit}+2}$ tokens. By Observation \ref{obs1}, for any set $S \in \mathcal{S}_{i_{crit}}$, we know that $|Cov(S)| < c(S) \cdot \beta^{i_{crit}+2}$. Thus: 
$$\mathcal{L}_{i_{crit}} = \sum_{S \in \mathcal{S}_{i_{crit}}} |Cov(S)| < \sum_{S \in \mathcal{S}_{i_{crit}}} c(S) \cdot \beta^{i_{crit}+2} = \mathcal{C}_{i_{crit}} \cdot \beta^{i_{crit}+2}$$

\noindent Therefore, level $i_{crit}$ has obtained more tokens than the number of elements covered at level $i_{crit}$, $\mathcal{L}_{i_{crit}}$, by using the tokens obtained by \emph{only} $\frac{1}{2} \cdot \frac{\epsilon}{\beta} \cdot \mathcal{C}_{i_{crit}} \cdot \beta^{i_{crit}}$ elements leaving the level, and we save the rest. The basis of the induction holds.

For the induction step, assume the claim holds inductively for some $1 \leq j \leq i_{crit}$, and  prove it for $j-1$. Since level $i_{crit}$ is half-critical, $\mathcal{D}_{j}^{i_{crit}} \geq \frac{1}{2} \cdot \frac{\epsilon}{\beta} \cdot \mathcal{C}_{j}^{i_{crit}}$. Assume that $\mathcal{D}_j^{i_{crit}} = \frac{1}{2} \cdot \frac{\epsilon}{\beta} \cdot \mathcal{C}_j^{i_{crit}} + X$, for some $X \geq 0$. By the induction hypothesis, we have not yet used the tokens from leaving elements that caused the overhead spare of $X$ for $\mathcal{D}_j^{i_{crit}}$. These elements left from levels between $j$ and $i_{crit}$, thus the number of leaving elements necessary to cause this overhead of $X$ is at least $X \cdot \beta^j$, meaning the overhead of $X$ gives at least $X \cdot \beta^j \cdot \frac{2\beta^3}{\epsilon}$ additional tokens, which have not been used yet. In addition, the tokens obtained from elements leaving level $j-1$ have not been used yet. Since $\mathcal{D}_{j-1}^{i_{crit}} \geq \frac{1}{2} \cdot \frac{\epsilon}{\beta} \cdot \mathcal{C}_{j-1}^{i_{crit}}$, we can write:
$$\mathcal{D}_{j}^{i_{crit}} + \mathcal{D}_{j-1} = \frac{1}{2} \cdot \frac{\epsilon}{\beta} \cdot (\mathcal{C}_{j}^{i_{crit}} + \mathcal{C}_{j-1}) + Y,$$

\noindent where $Y \geq 0$, and we want to prove that the tokens obtained from this overhead of $Y$ are not needed for level $j-1$ (and we can save them for later). Since $\mathcal{D}_j^{i_{crit}} = \frac{1}{2} \cdot \frac{\epsilon}{\beta} \cdot \mathcal{C}_j^{i_{crit}} + X$, we get that:
$$\frac{1}{2} \cdot \frac{\epsilon}{\beta} \cdot \mathcal{C}_j^{i_{crit}} + X + \mathcal{D}_{j-1} = \frac{1}{2} \cdot \frac{\epsilon}{\beta} \cdot (\mathcal{C}_{j}^{i_{crit}} + \mathcal{C}_{j-1}) + Y,$$

\noindent and thus:
$$\mathcal{D}_{j-1} = \frac{1}{2} \cdot \frac{\epsilon}{\beta} \cdot \mathcal{C}_{j-1} - X + Y$$

\noindent Level $j-1$ has obtained $\mathcal{D}_{j-1} \cdot \beta^{j-1} \cdot \frac{2\beta^3}{\epsilon}$ tokens. Thus, it obtained $(\frac{1}{2} \cdot \frac{\epsilon}{\beta} \cdot \mathcal{C}_{j-1} - X + Y) \cdot \beta^{j-1} \cdot \frac{2\beta^3}{\epsilon}$ tokens, meaning $\beta^{j+1} \cdot \mathcal{C}_{j-1} - \frac{2\beta^{j+2}}{\epsilon} \cdot X + \frac{2\beta^{j+2}}{\epsilon} \cdot Y$ tokens. Since the overhead of $X$ gave us $\geq \frac{2\beta^{j+3}}{\epsilon} \cdot X$ tokens, we have a total of $\geq \beta^{j+1} \cdot \mathcal{C}_{j-1} + \frac{2\beta^{j+2}}{\epsilon} \cdot Y$ tokens. By Observation \ref{obs1}, for any set $S \in \mathcal{S}_{j-1}$, we know that $|Cov(S)| < c(S) \cdot \beta^{j+1}$. Thus: 

$$\mathcal{L}_{j-1} = \sum_{S \in \mathcal{S}_{j-1}} |Cov(S)| < \sum_{S \in \mathcal{S}_{j-1}} c(S) \cdot \beta^{j+1} = \mathcal{C}_{j-1} \cdot \beta^{j+1}$$

\noindent Therefore, we have enough tokens to distribute one for each element covered at level $j-1$, by using only the tokens obtained by elements leaving levels $j-1$ to $i_{crit}$ that were necessary to raise $\mathcal{D}_{j-1}^{i_{crit}}$ to \emph{exactly} $\frac{1}{2} \cdot \frac{\epsilon}{\beta} \cdot \mathcal{C}_{j-1}^{i_{crit}}$, without using the overhead of $Y$. The induction step holds.

Applying the induction claim in the particular case $j = 0$, we conclude that we can pay all elements covered up to level $i_{crit}$ (one token each) by using tokens from elements that left levels between $0$ and $i_{crit}$ that were necessary to raise $\mathcal{D}_{0}^{i_{crit}}$ to $\frac{1}{2} \cdot \frac{\epsilon}{\beta} \cdot \mathcal{C}_{0}^{i_{crit}}$.
By Equation (\ref{eq1half}), since level $i_{crit}$ is half-critical, $\mathcal{D}_{0}^{i_{crit}} \geq \frac{1}{2} \cdot \frac{\epsilon}{\beta} \cdot \mathcal{C}_{0}^{i_{crit}}$, thus we have enough tokens, and the lemma holds. 
\qed
\end{proof}

\begin{cor} \label{cor1}

The amortized number of element level changes per update step is at most $\frac{\beta^2}{\epsilon}(1+\frac{2\beta^3}{\epsilon})\cdot O(\log_{\beta}n) + \frac{2\beta^3}{\epsilon} + 1 = O(\epsilon^{-3}\cdot \ln n) = O_{\epsilon}(\ln n)$.

\end{cor}

\begin{proof}
Denote by $K$ the total number of insertions and deletions in the sequence. There are three ways in which an element can change its level. The first is local rises, the second is partial resets, and the third is when the element is inserted/deleted. The bound provided by Lemma \ref{lem3} covers element level changes due to local rises and partial resets. Clearly, only one element level change can occur upon insertion/deletion - the updated element. It follows that the total number of dominated level changes in the sequence is no greater than $\leq [\frac{\beta^2}{\epsilon}(1+\frac{2\beta^3}{\epsilon})\cdot O(\log_{\beta}n) + \frac{2\beta^3}{\epsilon} + 1] \cdot K$, which completes the proof. 
\qed
\end{proof}

\noindent \textbf{Remark.} In DS, upon deletion of edge $(u,v)$ only one vertex can change its dominated level - if $u$ dominated $v$ and now $v$ becomes dominated by its next highest-level dominating neighbor. Upon insertion of $(u,v)$, again only one vertex can change its dominated level, if $v$ joins $Dom(u)$. 

\bigskip

\noindent Now that we have bounded the number of element level changes, we are ready to bound the number of set level changes, via the following simple observation.
\begin{obs} \label{clm3}

The number of set level changes throughout the whole sequence is at most twice the number of element level changes.

\end{obs}

\begin{proof}
Consider a set $S$. We first argue that any element that joins or leaves $Cov(S)$, changes its level. Indeed, an element $e$ can join/leave $Cov(S)$ in one of three ways: a local rise, a partial reset, and an insertion/deletion. A local rise necessarily changes the level of all the elements that join $Cov(S)$. Insertion/deletion of $e$ also changes its level. Following a partial reset where $e$ joins $Cov(S)$, $e$ might remain in the same level it was before the reset. Recall though (see the remark following Observation \ref{obs2}) that we view any element that enters a partial reset in $\tilde{\mathcal{U}}$ as changing its level. 
Importantly, viewing such elements as changing their levels, even when they do not, does not change the validity of the proof of Lemma \ref{lem3} and Corollary \ref{cor1}--- we can actually pay for every such element to "artificially" change its level. Thus, we conclude that any element that joins or leaves $Cov(S)$, changes its level. Next, we claim that between any two consecutive level changes of the set $S$ (including the latter one), at least one element leaves/joins $Cov(S)$. Indeed, if $S$ changes its level due to a local rise, at least one element joins the new $Cov(S)$. If $S$ changes its level due to a partial reset, then if its new level is $>-1$, at least one element joins $Cov(S)$. If its new level after the reset is $-1$, meaning $Cov(S) = \emptyset$, then at least one element left $Cov(S)$ from the time it was previously created at a level $>-1$ to this point. We conclude that we can ``assign" to each level change of $S$ a specific element level change that occurred during or before this set level change. Since each element level change can be assigned to up to two set level changes (the set which it leaves its Cov set and the set which it joins its Cov set), the observation follows.
\qed
\end{proof}

\noindent We conclude that the amortized number of level changes, both for elements and for sets, is $O(\epsilon^{-3}\cdot \ln n) = O_{\epsilon}(\ln n)$.

\subsection{$\Omega_{\epsilon}(\ln n)$ Amortized Element Level Changes} \label{omega}

\noindent In this section we show that the $O_{\epsilon}(\ln n)$ upper bound on the amortized number of element level changes is tight, meaning that exists such an input sequence of insertions and deletions of elements such that the amortized number of element level changes is $\Omega_{\epsilon}(\ln n)$. For the sake of simplicity and intuition we will consider $\beta = \sqrt{2}$ for this example, and we note that one can easily generalize for any $1 < \beta < 2$. 

The example is incremental, and assume that $n$ is a power of $2$, namely $n = 2^q$ for some natural $q \geq 2$. In the example all sets have a unit cost. We insert all $n$ elements in \emph{batches} of four, where the first four elements are denoted by the batch $B_1$ (a collection of elements which contains four elements), the second four by $B_2$, and so on until $B_{\frac{n}{4}}$. There are $(\frac{n}{2} - 1)$ sets, and each element can be covered by exactly $(\log_2(n) - 1)$ sets. All four elements within a certain batch of inserted elements can be covered by the same sets. The batch $B_i$ can be covered by all sets $S_j$ such that $j = i, \frac{n}{4} + \lfloor \frac{i+1}{2} \rfloor, \frac{3n}{8} + \lfloor \frac{i+3}{4} \rfloor, \ldots, \frac{n}{2} - 1$. An example for $n=32$ can be seen in Figure \ref{fig1}.

We begin by inserting the four elements of $B_1$. If the first inserted element joins $Cov(S_1)$ (arbitrary choice), then after the four insertions $S_1$ will cover all four elements at level $4$ (local rise after the second and fourth insertions to levels $2$ and $4$ respectively). Next, we insert the four elements of $B_2$. Again, if the first one joins $Cov(S_2)$ (arbitrary choice), then after the third insertion $S_2$ covers the three elements at level $2$. The fourth inserted element joins $Cov(S_2)$ as well, but then immediately a local rise occurs, since $|N_5(S_{\frac{n}{4} + 1})| \geq \beta^{6}$ (in the example given in Figure \ref{fig1}, $|N_5(S_9)| = 8 = \beta^{6}$). So following these eight insertions, $S_{\frac{n}{4} + 1}$ covers all eight elements at level $6$. Notice that the elements in $B_2$ were not covered at level $4$, and all four jumped from level $2$ to level $6$. We continue with $B_3$ and $B_4$ in the same manner. This time, after the insertion of the fourth element in $B_4$, we have $|N_5(S_{\frac{3n}{8} + 1})| \geq \beta^{8}$ (in the example given in Figure \ref{fig1}, $|N_7(S_{13})| = 16 = \beta^{8}$). So following these sixteen insertions, $S_{\frac{3n}{8} + 1}$ covers all sixteen elements at level $8$. Notice that the elements in $B_3$ were not covered at level $6$, and all four jumped from level $4$ to level $8$. In addition, the elements in $B_4$ were not covered at level $4$ or $6$, and all four jumped from level $2$ to level $8$. This process continues in the same manner until all elements are inserted.

Consider the bit representation of $i$ for each batch $B_i$. All batches $B_i$ in which the LSB of $i$ is $1$ (half of the batches), the four elements in $B_i$ are covered at level $4$ at some point. All batches $B_i$ in which the two LSB's of $i$ represent a decimal number between $1$ and $2$ (half of the batches), the four elements in $B_i$ are covered at level $6$ at some point. All batches $B_i$ in which the three LSB's of $i$ represent a decimal number between $1$ and $4$ (half of the batches), the four elements in $B_i$ are covered at level $8$ at some point. All batches $B_i$ in which the four LSB's of $i$ represent a decimal number between $1$ and $8$ (half of the batches), the four elements in $B_i$ are covered at level $10$ at some point, and so on. Overall there are $\frac{n}{4}$ batches. The highest level is $(\log_{\sqrt{2}}n) = (2 \cdot \log_2 n)$. The relevant levels are $4$ to $(2 \cdot \log_2 n)$, in increments of $2$. Meaning, there are $\frac{2 \cdot \log_2 n - 4}{2} + 1 = (\log_2 n - 1)$ levels, where for each level, all of the elements in half of the batches ($\frac{n}{2}$ elements) at some point were covered at that level. Therefore, the total number of element level changes is at least $\frac{n}{2} \cdot (\log_2 n - 1)$. Since the total number of update steps $k$ is clearly equal to $n$, we get that there are at least $\frac{k}{2} \cdot (\log_2 n - 1) = \Omega(k \cdot \ln n)$ element level changes, meaning $\Omega(\ln n)$ amortized element level changes.

\begin{figure}
  \center{\includegraphics[scale=0.35]{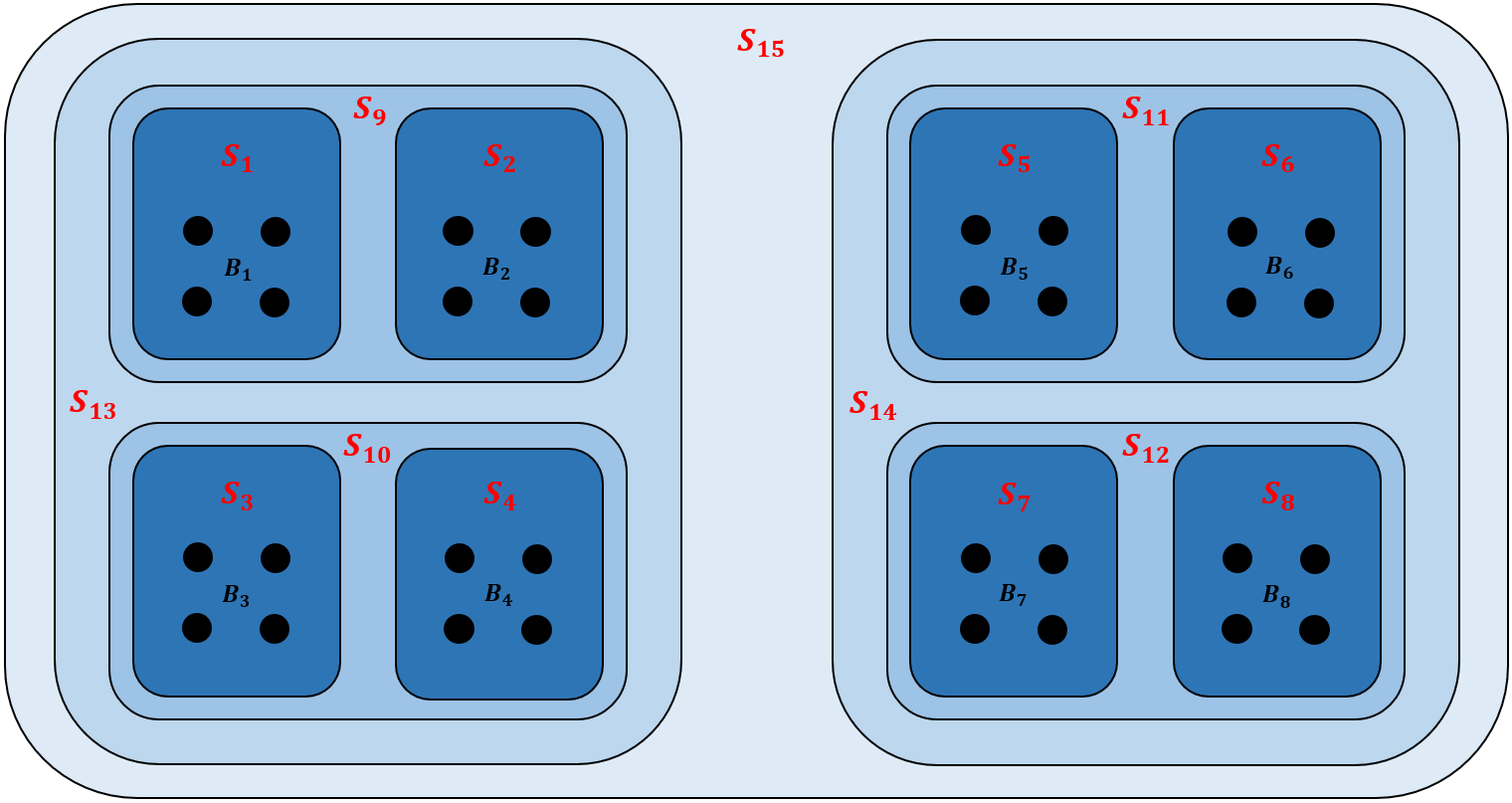}}
  \caption{\label{fig1} Element assignments to sets for $n=32$.}
\end{figure}

For the DS problem, the example is somewhat similar but decremental, and can be seen in Figure \ref{fig2}. We consider a bipartite graph, where on one side we have $2^q$ \emph{batches} of vertices, denoted by $B_1, B_2, \ldots, B_{2^q}$ and on the other side $2^{q+1}-1$ vertices denoted by $v_1,v_2,\ldots,v_{2^{q+1}-1}$. Just as in SC, the vertices within a certain batch will have the same set of neighbors. The batch $B_i$ is connected to all vertices $v_j$ such that $j = i, 2^q + \lfloor \frac{i+1}{2} \rfloor, \frac{3}{2}\cdot 2^q + \lfloor \frac{i+3}{4} \rfloor, \ldots, 2^{q+1} - 1$. In DS, since a vertex can dominate itself, the sizes of the batches are different. Consider batch $B_i$ and the binary representation of $i$. Denote by $x_i$ the least significant bit in $i$ that is $1$, where the LSB is $x_i=0$. Then the size of $B_i$ is $x_i+3$. For example, for odd $i$ the size of $B_i$ is three. For $i=2,6,10,\ldots$ the size of $B_i$ is four. For $i=4,12,20,\ldots$ the size of $B_i$ is five, and so on. We get that there are $2^{q+2}-1$ vertices in the batches. Thus, $n=3 \cdot 2^{q+1} - 2$, and there are $m = (2^{q+2}-1) \cdot (q+1)$ edges, meaning $m = \Theta(n\log n)$. Again for the sake of simplicity and intuition we will consider $\beta = \sqrt{2}$ for this example, and we note that one can easily generalize for any $1 < \beta < 2$. 

Initially, $v_{2^{q+1}-1}$ (the ``root" which is connected to all batches) dominates all vertices in the batches and itself. In addition, all vertices $v_j$ dominate only themselves. We begin by deleting all edges connecting $v_{2^{q+1}-1}$ to the vertices in $B_1$. After each deletion, the undominated vertex in $B_1$ chooses to be dominated by $v_1$ (arbitrary choice). Once all three edges are deleted, $v_1$ dominates itself and the three vertices in $B_1$ and performs a local rise to level $4$. We continue by deleting all edges connecting $v_{2^{q+1}-1}$ to vertices in $B_2, B_3, \ldots$ in order. Once three of the four edges to $B_2$ are deleted, $v_2$ performs a local rise to level $4$ and dominates itself and the three vertices in $B_2$. Once the fourth edge is deleted, the fourth vertex changes to level $4$ as well (dominated by $v_2$), but now $v_{2^q+1}$ has eight ($\beta^6$) neighbors at level less than $5$, thus all eight vertices ($v_{2^q+1}$, the three in $B_1$ and the four in $B_2$) join the local rise of $v_{2^q+1}$ to level $6$. Similarly, after the deletion of the edge connecting $v_{2^{q+1}-1}$ to the last vertex in $B_4$, $v_{\frac{3}{2} \cdot 2^q+1}$ dominates all vertices in $B_i$ ($i=[1,4]$) at level $8$, and so on. Eventually, after deleting all edges adjacent to $v_{2^{q+1}-1}$, we get that $v_{2^{q+1}-3}$ and $v_{2^{q+1}-2}$ will each dominate $2^{q+1}$ and $2^{q+1}+1$ vertices respectively (including themselves), at level $2q+2$. We get that there are $\Theta(2^q)$ deletions, and $\Theta(q \cdot 2^q)$ dominated level changes (similarly to the SC analysis), meaning $\log k$ amortized level changes, after the $k= \Theta(2^q)$ deletions. 

Now, $v_{2^{q+1}-3}$ is the new ``root" of batches $B_1$ through $B_{2^{q-1}}$ and $v_{2^{q+1}-2}$ is the new ``root" of batches $B_{2^{q-1}+1}$ through $B_{2^q}$. Thus, we can continue the same process as before, meaning deleting all edges connecting $v_{2^{q+1}-3}$ to batches $B_i$ (for $1 \leq i \leq 2^{q-1}$). The vertices in these batches will slowly climb up the levels, until about half are dominated by $v_{2^{q+1}-7}$ and the other half are dominated by $v_{2^{q+1}-6}$, and this process continues. Similarly, we delete all edges connecting $v_{2^{q+1}-2}$ to batches $B_i$ (for $2^{q-1}+1 \leq i \leq 2^q$). The vertices in these batches will slowly climb up the levels, until about half are dominated by $v_{2^{q+1}-5}$ and the other half are dominated by $v_{2^{q+1}-4}$, and this process continues. In the first iteration (where we only had one ``root", $v_{2^{q+1}-1}$) all vertices in the batches climbed up the levels two by two until they reached $2q+2$. In the second iteration (two ``roots") all vertices in the batches climbed up the levels two by two until they reached $2q$. In the third iteration all vertices in the batches climb up to $2q-2$, and so on. Following the last iteration, the vertices in the batches dominate themselves at level $0$. Therefore, at each iteration each vertex in the batches climbs on average $\Omega(q)$ levels. Since there are $\Theta(q)$ iterations, each vertex in the batches has $\Omega(q^2)$ dominated level changes. Thus, there is a total of $\Omega(q^2 \cdot 2^q)$ dominated level changes. The total number of deletions is the number of initial edges, $(2^{q+2}-1) \cdot (q+1) = O(q \cdot 2^q)$. Therefore, the amortized number of dominated level changes is $\Omega(q)$, which is $\Omega(\ln n)$.

\begin{figure}
  \center{\includegraphics[scale=0.42]{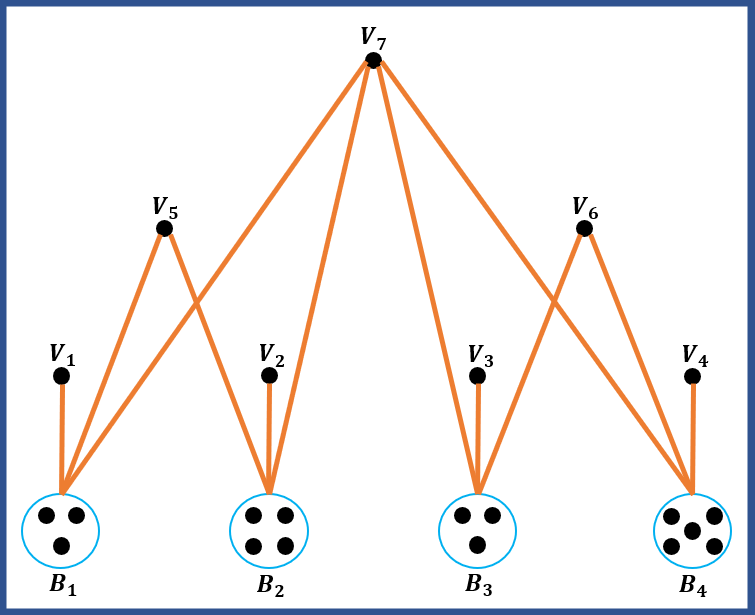}}
  \caption{\label{fig2} Graph for $q=2$.}
\end{figure}

\subsection{From Set and Element Level Changes to Update Time} \label{updatetime}
In this section we demonstrate that the amortized update time of the algorithm exceeds the amortized number of level changes by a factor of $O_\eps(f)$, yielding the required update time. For each set $S$ (with cost $c(S)$), our algorithm will maintain:

\begin{itemize}
\item The elements in $S$ maintained as a doubly-linked list. 
\item The level $l_{cov}(S)$ of $S$, where we remind that $l_{cov}(S) = -1$ if $Cov(S) = \emptyset$.
\item The Cov set $Cov(S)$ of $S$ (which is $\emptyset$ if $S$ isn't covering and nonempty otherwise), maintained as a doubly-linked list.
\item The counters $N_j(S)$ and auxiliary counters $N^j(S)$ and $C_S$ as described below. 
\end{itemize}
\noindent For each element $e$, our algorithm will maintain:
\begin{itemize}
\item The up to $f$ sets that can cover $e$, $\mathcal{F}_e$, maintained as a doubly-linked list.
\item The level $l(e)$ of $e$.
\end{itemize}

\noindent In addition, for each level $j$, our algorithm will maintain doubly-linked lists of $\mathcal{L}_j$, $\mathcal{D}_j$, $\mathcal{C}_j$, $\mathcal{E}_j$ and $\mathcal{S}_j$. We also maintain $\mathcal{D}$, $\mathcal{C}$ and $\mathcal{S}$ as doubly-linked lists. The maintenance of all these lists can be efficiently carried out in the obvious way, and is required for an efficient implementation of the partial reset procedure, given at the end of this section.

\bigskip

\noindent {\bf Remark.} All the aforementioned doubly-linked lists should also hold appropriate pointers to elements, to allow for insertions and deletions of elements in constant time.

\bigskip

\noindent We have bounded the amortized number of level changes, both for elements and for sets, by $O(\epsilon^{-3}\cdot \ln n) = O_{\epsilon}(\ln n)$. The more important of the two is the number of element level changes. Every time there is an element level change by the algorithm, say of element $e$, we will scan all (at most $f$) sets in $\mathcal{F}_e$ to update the required data structures, as described below; that would require amortized update time of $O_{\epsilon}(f \cdot \ln n)$. We do so since it is crucial that every set $S$ will have up-to-date information concerning the level of the elements that it could cover.
On the other hand, following a set level change, we will not scan {\em all} elements that the set can cover, since there is no need to keep up-to-date information concerning this level change. More importantly, we cannot afford to scan all elements in a set, since this could be $\Omega(n)$.
Instead, we will scan only the elements in the old and new Cov sets of $S$ (for updating the respective Cov sets) which requires time linear in the sizes of the Cov sets - recalling that we view all these elements as changing their level, scanning them and updating the respective Cov sets requires time linear in the number of element level changes.

\begin{obs} \label{upper}
The total runtime due to set level changes is upper bounded by that due to element level changes.
\end{obs}

\noindent Given Observation \ref{upper}, we will mostly focus on the runtime incurred due to element level changes, including the runtime due to local rises and partial resets. 

\subsubsection{Maintaining the $N_{j}(S)$ Counters}
Our goal now is to prove that the \emph{ChangeLevel} procedure (see Algorithm \ref{algcl}) can be implemented in $O(\epsilon^{-2} \cdot f)$ amortized time. 
For every set $S$ we maintain a counter $C_S$, which increments every time an element in $S$ changes its level. 
(We shall reset this counter every $\Theta(n^2)$ steps, so that it does not occupy more than $O(\ln n)$ bits, regardless of the length of the update sequence). We also maintain $\log_{\beta}(n)+3$ (see Section \ref{rel}) additional counters for each set $S$, denoted $N^j(S)$ (for each level $j$ that is a relevant level of $S$), which represents the number of elements in $S$ that are at level {\em exactly} $j$. Upon a level change of an element $e$ from $j'$ to $j$, we will update three counters for each $S \in \mathcal{F}_e$: $C_S$ (increment by one), $N^{j'}(S)$ (decrement by one if $j'$ is a relevant level of $S$) and $N^{j}(S)$ (increment by one if $j$ is a relevant level of $S$). Since $|\mathcal{F}_e| \leq f$, each element level change takes $O(f)$ time. Since the amortized number of element level changes is $O_{\epsilon}(\ln n)$ (see Corollary \ref{cor1}), the amortized update time will be $O_{\epsilon}(f \cdot \ln n)$.

However, this by itself does not suffice, since we need to maintain the $N_j(S)$ counters. Keeping the $N_j(S)$ counters 
(for all relevant levels) up-to-date following every update step would naively give rise to an amortized update time of  $O_{\epsilon}(f \cdot \ln^2 n)$, which would be too costly. Instead, we will maintain {\em approximate} estimations of these counters, allowing an additive error of $\epsilon \cdot \beta^j$ for $|N_j(S)|$. Such an additive error is negligible, and has no significant effect on the analysis detailed in the previous sections. (In more detail, one has to use a buffer of size $\beta^2$ instead of $\beta$ in the definitions of dirty in order to absorb these additive errors, which will cause the overall approximation guarantee to increase by a multiplicative factor of $(1+O(\epsilon))$. We can bring down the approximation factor to the same value as before by a straightforward scaling argument.) 
{The approximation factor argument provided in Section \ref{approx} takes into account the additive error to the counters $|N_j(S)|$.} 
Since each increment of $C_S$ could change $|N_j(S)|$ (for any $j$) by only one, it suffices to update $N_j(S)$ after $\epsilon \cdot \beta^j$ increments to $C_S$. 
We partition all of the relevant levels into (roughly) equal sized sets, and denote the sets $Z_1,Z_2,\ldots$. For simplicity assume that the relevant levels are from $0$ to $\log_{\beta}n+2$ (where we can simply add $j_{min}^S$ to each level so that we obtain the relevant levels). In $Z_1$ we have all the levels $j$ such that $0 \leq j \leq \lfloor \log_{\beta}(\frac{2}{\epsilon}) \rfloor $. In $Z_i$ (for any $i>1$) we have all the levels $j$ such that $\lfloor \log_{\beta}(\frac{2^{i-1}}{\epsilon}) \rfloor + 1 \leq j \leq \lfloor \log_{\beta}(\frac{2^i}{\epsilon}) \rfloor $. The size of the first set is about (give or take one due to the floor value) $\log_{\beta}(\frac{2}{\epsilon})$. The sizes of all other sets are about ($\pm 1$) $\log_{\beta}(\frac{2^{i}}{\epsilon}) - \log_{\beta}(\frac{2^{i-1}}{\epsilon}) = \log_{\beta}(2)$. 
We consider the binary representation of $C_S$, and associate the LSB to the set $Z_1$, the second LSB is associated to $Z_2$, etc.
We analyze the incrementation of $C_S$ via the standard amortized analysis of binary counter incrementation,
which implies not only a constant amortized number of bit flips per increment step, but also the two properties ({\bf P1} and {\bf P2}) provided next. 
In particular, for each bit $i$ that flips, we update all the $N_j(S)$ counters such that $j \in Z_i$; we describe below how to efficiently carry out this update of the $N_j(S)$ counters.
\begin{enumerate}
	\item \textbf{P1:} Once we update $N_j(S)$ for some $j$, we also update $N_{j'}(S)$ for any level $j'<j$.
	\item \textbf{P2:} For any $Z_i$, we update $N_j(S)$ for any $j \in Z_i$ every $2^{i-1}$ increments of $C_S$.
\end{enumerate}

\noindent Since $j > \log_{\beta}(\frac{2^{i-1}}{\epsilon})$ for any $i > 1$ and $j \in Z_i$, we get that $\epsilon \cdot \beta^j > 2^{i-1}$. From {\bf P2}, since we update this counter $N_j(S)$ every $2^{i-1}$ increments of $C_S$, we update it at least once every $\epsilon \cdot \beta^j$ increments of $C_S$, implying that the additive error is at most $\epsilon \cdot \beta^j$, as required.
(For $i=1$, we update $N_j(S)$ for $j \in Z_1$ every increment of $C_S$, so there is no additive error whatsoever in this case.) 
Since the amortized number of bit flips is constant, and each $Z_i$ is of size $\leq \log_{\beta}(\frac{2}{\epsilon}) < \frac{4}{\epsilon^2}$, the amortized number of updates to all the $N_j(S)$ counters per increment of $C_S$ is $O(\epsilon^{-2})$.
To conclude, upon each element level change, we increment/decrement at most $3f$ counters, $f$ of which are the $C_S$ counters, and the amortized number of updates to $N_j(S)$ that they trigger is $O(\epsilon^{-2})$. Thus, the amortized cost required by each element level change is $O(\epsilon^{-2} \cdot f)$, disregarding the time needed for computing the $N_j(S)$ counters.

Finally, we describe how to update the $N_j(S)$ counters. Recall that we always have the exact up-to-date values of all the $N^j(S)$ counters. Notice that $N_j(S) = N_{j-1}(S) + N^{j-1}(S)$, for any $j \ge 1$. From \textbf{P1}, we know that upon each update of $N_j(S)$, we update all previous ones $N_{j-1}(S),\ldots, N_1(S)$ as well. Thus, we can first update $N_1(S) = N^0(S)$ (or $N_{j_{min}^S+1}(S) = N^{j_{min}^S}(S)$). Then, iteratively compute $N_2(S) = N_{1}(S) + N^{1}(S), N_3(S) = N_{2}(S) + N^{2}(S), \ldots, N_k(S) = N_{k-1}(S) + N^{k-1}(S)$, where $k$ is the highest level in $Z_q$ and $q$ is the MSB of $C_S$ that has flipped. By {\bf P1}, by updating the $N_j(S)$ values in this way consecutively from index $1$ through $k$, we have guaranteed that all the $N_j(S)$ values until index $k$ are up-to-date. The total running time for computing $N_j(S)$ in this way
over all $j = 1, \ldots, k$ is $O(k)$, hence we get a constant amortized cost. See Algorithm \ref{alg5} for the pseudo-code of the update procedure of the $N_j(S)$ counters.

\begin{algorithm}
\caption{UpdateLevelChange $(e, j', j)$ (following an element level change of $e$ from $j'$ to $j$)} \label{alg5}
\begin{algorithmic}
\FOR {$S \in \mathcal{F}_e$}
\STATE increment $C_S$ (and consider the corresponding binary vector)
\STATE decrement $N^{j'}(S)$ if $j'$ is a relevant level of $S$
\STATE increment $N^j(S)$ if $j$ is a relevant level of $S$
\STATE Let $q$ be most significant bit of $C_S$ that flips
\FOR {$i=1$ to $q$}
\FOR {$j \in Z_i$ (in increasing order of j)}
\STATE $N_j(S) = N_{j-1}(S)+N^{j-1}(S)$
\ENDFOR
\ENDFOR
\ENDFOR
\end{algorithmic}
\end{algorithm}

\noindent It follows that the overall \emph{amortized} cost of each element level change is $O(\epsilon^{-2} \cdot f)$.
{As mentioned, the amortized cost due to a set level change is dominated by that of an element level change.}
We have shown (Corollary \ref{cor1}) that the amortized number of all level changes per update step is $O(\epsilon^{-3} \cdot \ln n)$,
hence we conclude with the following corollary:

\begin{cor} \label{finalrt}
The amortized update time due to element level changes is $O(\eps^{-5} \cdot f \cdot \ln n) = O_\eps(f \cdot \ln  n)$.
\end{cor}

\subsubsection{Local Rise}
If a set $S$ is $j$-PD, we perform a local rise of $S$ to level $j+1$. 
If there are multiple levels $j$ for which $S$ is $j$-PD, recall that the highest level should be taken. 
Determining if $S$ is $j$-PD for a certain level $j$ boils down to the maintenance of the respective $N_j(S)$ counter. 
Following any update step, for any counter $N_j(S)$ that is updated due to that step (as per Algorithm \ref{alg5}), 
we can determine if $S$ is $j$-PD in constant time; if there are multiple such levels $j$, finding the highest among them for which $S$ is $j$-PD takes time linear in the number of $N_j(S)$ counters that are updated due to that update step. 
Thus the overall time incurred by all creations of Cov sets due the local rise procedure is asymptotically the same as that of maintaining all the $N_j(S)$ counters. 
Besides the creation time of the Cov sets, for each element $e$ in the newly created set $Cov(S)$ (recall that we view all of them as changing their levels), we need to scan all sets $S' \in \mathcal{F}_e$ and update their $N_j(S')$ counters; however, we have already argued that the amortized update time due to the updates of all such counters does not exceed $O_\eps(f \cdot \ln  n)$.

\subsubsection{Partial Reset}
Our goal now is to prove that the execution of partial resets is not the bottleneck of the algorithm, and that it is dominated by the update time due to element level changes (see Corollary \ref{finalrt}). 
Recall that a partial reset amounts to running the {\em static greedy} set cover algorithm on the sets $\tilde{\mathcal{U}}$ and $\tilde{\mathcal{F}}$
as input, which are induced by the set cover $\mathcal{S}$ and a half-critical level $i_{crit}$. When performing a reset, we must first find a half-critical level once the system is dirty, then compute the sets $\tilde{\mathcal{U}}$ and $\tilde{\mathcal{F}}$, and finally apply the static greedy algorithm on these sets.

\bigskip

\noindent \textbf{Computing a half-critical level $i_{crit}$:} 
Since there are $\Theta(\log_{\beta} (n \cdot C))$ possible levels, naively finding a half-critical level could take $\Omega(\log_{\beta}^2 (n \cdot C))$ time, assuming the counters $\mathcal{C}$, $\mathcal{D}$, $\mathcal{C}^j_j$ and $\mathcal{D}^j_j$ are always up-to-date; as mentioned above, these counters can be maintained efficiently following insertions and deletions of elements by our algorithm in the obvious way. 

We will perform a \emph{global reset} every once in a while, which means running the static greedy algorithm on the entire system. Say at time step $t$ we perform a global reset. At this time, there are $n_t$ active elements. The next global reset will be done at time step $t + n_t$. Again, at this time say there are $n'_t$ elements, so the next global reset will be at time step $t + n_t + n'_t$, and so on. We define an \emph{idle set} to be a set in the SC without any elements, meaning set $S$ is an idle set if $|Cov(S)| = 0$ and $l_{cov}(S) \geq 0$. Notice that following a global reset, there are no idle sets. Moreover, the number of element level changes is an upper bound to the number of idle sets (since the last global reset), since for a set to become idle at least one element must leave its Cov set, therefore changing its level. Consider a global reset at time step $t$. By Corollary \ref{cor1}, in the $n_t$ update steps after $t$, the number of element level changes is $O(\epsilon^{-3} \cdot \ln n \cdot n_t)$. Therefore, the number of idle sets at $t+n_t$ is $O(\epsilon^{-3} \cdot \ln n \cdot n_t)$. If we consider a token scheme where each update step gives $\Theta(f)$ tokens, then at the global reset at time $t+n_t$ we have enough tokens to execute the global reset. Indeed, the running time of a global reset is $O(|\tilde{\mathcal{U}}| \cdot f + \log_{\beta}n)$ (see Corollary \ref{rstime}), where $|\tilde{\mathcal{U}}| \leq 2n_t$, since at time $t$ there were $n_t$ elements and after $n_t$ update steps there could only be $n_t$ more elements (if all updates were insertions). Thus, at time step $t+n_t$ we have $\Theta(n_t \cdot f)$ tokens, which is $\Omega (|\tilde{\mathcal{U}}| \cdot f)$. Since the amortized running time will be $O_{\epsilon}(f \cdot \ln n)$, we can afford the extra addition of $\Theta(f)$ tokens per update step, and an extra $O(\log_{\beta}n)$ even every update step, without changing the total running time. We have shown therefore that these global resets have no effect on the running time, and the number of idle sets is always bounded by $O(\epsilon^{-3} \cdot \ln n \cdot n)$. The number of sets that are not idle is bounded by $n$, since each of these sets covers at least one element, and no element is covered by more than one set. Thus, the total number of sets in the SC is $O(\epsilon^{-3} \cdot \ln n \cdot n)$. We record this observation for future use:

\begin{obs} \label{global}
The total number of sets in the SC at any given time step is $O(\epsilon^{-3} \cdot \ln n \cdot n)$.
\end{obs} 

\noindent Once the system is dirty, we want to find a half-critical level to do a partial reset up to. Let $b$ denote the lowest level in which a set in the SC is at in a certain point in time, and let $u = b + 4 \log_{\beta}n$. We search for a half-critical level between the levels $b$ and $u$ to do a reset up to, where a few challenges/questions arise:

\begin{enumerate}
\item The naive running time of this procedure could be $\Omega(\log_{\beta}^2 n)$, which is still more than what we are aiming for. We would like to find such a level more efficiently, in $O(\log_{\beta} n)$ time, which would allow us to search for a half-critical level every update step, without changing the asymptotic update time.
\item Once we find this half-critical level and perform a reset up to it, we would like to show that there is no need to immediately search for another one, meaning the system is not dirty. This means that between any two resets we will have an update step, which can ``recharge" this procedure.
\item Once the system is dirty we know that exists a half-critical level, but we need to prove that indeed exists one up to level $u$.
\end{enumerate}

\noindent We begin with the third question, by proving a few useful claims first.

\begin{cl} \label{impcd}
For any $i \leq j$, $\mathcal{D}_i^j < \beta \cdot \mathcal{C}_i^j$.
\end{cl}

\begin{proof}
Consider a set $S$ such that $l_{cov}(S) = j$. Since the initial $|Cov(S)|$ is less than $c(S) \cdot \beta^{j+1}$, the total contribution to $\mathcal{D}_j$ of elements leaving $S$ is $<\frac{c(S) \cdot \beta^{j+1}}{\beta^j} = c(S) \cdot \beta$. Summing up on all sets at level $j$, we obtain:

$$ \mathcal{D}_j < \sum_{S \in \mathcal{S}_j} c(S) \cdot \beta = \beta \cdot \mathcal{C}_j, $$

\noindent and summing up on all levels $i$ to $j$ we get:

$$ \mathcal{D}_i^j < \beta \cdot \mathcal{C}_i^j $$

\qed
\end{proof}

\begin{cl} \label{epsfive}
$\mathcal{C}_{u+1}^{\infty} < \epsilon^5 \cdot \mathcal{C}$.
\end{cl}

\begin{proof}
Consider a set $S'$ at level $>u$. The cost of $S'$ must satisfy $c(S') < \frac{n}{\beta^u}$ (no more than $n$ elements). Plugging in $u = b + 4 \log_{\beta}n$, we get $c(S') < \frac{1}{\beta^b \cdot n^3}$. Denote by $S_b$ a set covering at level $b$ and denote by $c_b$ its cost. Denote by $x_b$ the size of $Cov(S_b)$ when it was created, and we know it was at least one. We also know that $x_b < c_b \cdot \beta^{b+1}$, otherwise $Cov(S_b)$ would have been created at a level higher than $b$. Thus, $c_b > \frac{1}{\beta^{b+1}}$. Therefore, $c(S') < \frac{\beta \cdot c_b}{n^3}$. Since the total number of sets in the SC is $O(\epsilon^{-3} \cdot \ln n \cdot n)$ by Observation \ref{global}, the total cost of all sets above level $u$ is $<\frac{\beta \cdot c_b}{n^3} \cdot O(\epsilon^{-3} \cdot \ln n \cdot n) < c_b\cdot O(\epsilon^{-3}) \cdot \frac{1}{n} \leq \mathcal{C} \cdot O(\epsilon^{-3}) \cdot \frac{1}{n}$. Since we assume $n$ is large enough such that $n = \Omega(\epsilon^{-9})$, we get that the total cost of all sets above level $u$ is $<\epsilon^5 \cdot \mathcal{C}$.
\qed
\end{proof}

\begin{cl} \label{negl1}
Once the system is dirty, exists a half-critical level up to level $u$. 
\end{cl}

\begin{proof}
We will show that once the system is dirty, $\mathcal{D}_0^u \geq \frac{1}{2} \cdot \frac{\epsilon}{\beta} \cdot \mathcal{C}_0^u$. Then, we can use Claim \ref{lem1}, with $u$ taken as the highest level, and conclude that a half-critical level exists up to level $u$. Since the system is dirty, we can write $\mathcal{D} = \mathcal{D}_0^u + \mathcal{D}_{u+1}^{\infty} \geq \frac{\epsilon}{\beta} \cdot (\mathcal{C}_0^u + \mathcal{C}_{u+1}^{\infty})$. From here we can write:

\begin{align*}
\mathcal{D}_0^u & \geq \frac{\epsilon}{\beta} \cdot (\mathcal{C}_0^u + \mathcal{C}_{u+1}^{\infty}) - \mathcal{D}_{u+1}^{\infty} \\ & > \frac{\epsilon}{\beta} \cdot (\mathcal{C}_0^u + \mathcal{C}_{u+1}^{\infty}) - \beta \cdot \mathcal{C}_{u+1}^{\infty} \\ & > \frac{\epsilon}{\beta} \cdot (\mathcal{C}_0^u + \mathcal{C}_{u+1}^{\infty}) - \beta \cdot \epsilon^5 \cdot \mathcal{C} \\ & = \frac{\epsilon}{\beta} \cdot (\mathcal{C}_0^u + \mathcal{C}_{u+1}^{\infty}) - \beta \cdot \epsilon^5 \cdot (\mathcal{C}_0^u + \mathcal{C}_{u+1}^{\infty}) \\ & = (\frac{\epsilon}{\beta} - \beta \cdot \epsilon^5) \cdot \mathcal{C}_0^u + (\frac{\epsilon}{\beta} - \beta \cdot \epsilon^5) \cdot \mathcal{C}_{u+1}^{\infty} \\ & > (\frac{\epsilon}{\beta} - \beta \cdot \epsilon^5) \cdot \mathcal{C}_0^u \\ & > \frac{1}{2} \cdot \frac{\epsilon}{\beta} \cdot \mathcal{C}_0^u,
\end{align*}

\noindent where the second transition is by Claim \ref{impcd}, the third transition is by Claim \ref{epsfive}, the sixth for $\epsilon<0.75$, and the last for $\epsilon<0.65$.
\qed
\end{proof}

\noindent Now, we address the first question, and explain how to efficiently find a half-critical level up to level $u$, in $O(\log_{\beta}n)$ worst case running time. Notice that there could be several such levels, and we will find the \emph{highest} one (up to level $u$). Since a half-critical level must be half-dirty too by definition, we will first create a sorted list containing all half dirty levels between $b$ and $u$. Since we know $\mathcal{C}_i$ and $\mathcal{D}_i$ for any $i$, checking if level $i$ is half-dirty takes $O(1)$ time. Denote the lowest half-dirty level by $i_1$, the next by $i_2$, and so on. In the process of searching for the highest half-critical level up to $u$, we will maintain a list of \emph{dirty ranges}. Each dirty range is defined by two levels, $j$ and $i$, where $j<i$. For each range we will store four values: $j$, $i$, $\mathcal{D}_{j+1}^i$ and $\mathcal{C}_{j+1}^i$. Initially, the list of dirty ranges is empty. We begin by checking if $i_1$ is half-critical. Since we have the values of $\mathcal{D}_{i}$ and $\mathcal{C}_{i}$ for any $i$, we can add $\mathcal{D}_{i_1-1}$ and $\mathcal{C}_{i_1-1}$ to $\mathcal{D}_{i_1}$ and $\mathcal{C}_{i_1}$ respectively and check in $O(1)$ time if $\mathcal{D}_{i_1-1}^{i_1} \geq \frac{1}{2} \cdot \frac{\epsilon}{\beta} \cdot \mathcal{C}_{i_1-1}^{i_1}$. We continue in the same manner downwards while incrementing until we either reach level $b$ and then $i_1$ is half-critical, or we reach some level $j_1$ such that $\mathcal{D}_{j_1}^{i_1} < \frac{1}{2} \cdot \frac{\epsilon}{\beta} \cdot \mathcal{C}_{j_1}^{i_1}$. In the first case, we add $(b-1,i_1)$ to the list of dirty ranges. Since we just calculated $\mathcal{D}_{b}^{i_1}$ and $\mathcal{C}_{b}^{i_1}$, we have the four values we need for this dirty range. In the second case, we add $(j_1,i_1)$ to the list of dirty ranges. Since we calculated $\mathcal{D}_{j_1+1}^{i_1}$ and $\mathcal{C}_{j_1+1}^{i_1}$ just before reaching $j_1$, we have the four values we need for this dirty range. We now move on to the next half-dirty level, $i_2$. Again, we check if $i_2$ is half-critical in the same manner, by progressing downwards while incrementing, in $O(1)$ time for each progression. Eventually, we again either reach level $b$ and then $i_2$ is half-critical, or we reach some level $j_2$ such that $\mathcal{D}_{j_2}^{i_2} < \frac{1}{2} \cdot \frac{\epsilon}{\beta} \cdot \mathcal{C}_{j_2}^{i_2}$. In the first case we add $(b-1,i_2)$ to the list of dirty ranges, and in the second case we add $(j_2,i_2)$. Upon progressing downwards, if we reach the upper bound of a dirty range, then we will skip the dirty range and continue from the lower bound of that dirty range. We are able to skip a dirty range and continue our progression after it since we have stored $\mathcal{D}_{j+1}^i$ and $\mathcal{C}_{j+1}^i$ for each dirty range $(j,i)$, so we just need to add these values and continue. To justify the fact that we can skip a dirty range, say without loss of generality we are verifying if $i_2$ is half-critical. If we reach $i_1$, we already know that:

\begin{equation} \label{newcrit1}
\mathcal{D}_{i_1+1}^{i_2} \geq \frac{1}{2} \cdot \frac{\epsilon}{\beta} \cdot \mathcal{C}_{i_1+1}^{i_2}
\end{equation} 

\noindent Since $(j_1,i_1)$ is a dirty range (where $j_1$ could be $b-1$ too), for any $j_1+1 \leq j \leq i_1$ we have that:

\begin{equation} \label{newcrit2}
\mathcal{D}_{j}^{i_1} \geq \frac{1}{2} \cdot \frac{\epsilon}{\beta} \cdot \mathcal{C}_{j}^{i_1}
\end{equation}

\noindent Adding both sides of the inequalities (\ref{newcrit1}) and (\ref{newcrit2}) we get:

\begin{equation} \label{newcrit3}
\mathcal{D}_{j}^{i_2} \geq \frac{1}{2} \cdot \frac{\epsilon}{\beta} \cdot \mathcal{C}_{j}^{i_2}
\end{equation}

\noindent Therefore, there is no need to check for any $j_1+1 \leq j \leq i_1$, and we can skip from $i_1+1$ to $j_1$. Moreover, once we reach a dirty range we delete it from our list, since it will be contained in the new dirty range we are now dealing with. Thus, at all times the dirty ranges are disjoint. Since by Claim \ref{negl1} a half-critical level between $b$ and $u$ exists, and we find \emph{all} of them, we can return the highest one. See Algorithm \ref{alg6} for code and further details.

\begin{algorithm}
\caption{FindHighestHalfCriticalLevel $(b)$} \label{alg6}
\begin{algorithmic}
\STATE $u \leftarrow b + 4 \log_{\beta}n$
\STATE $k \leftarrow 0$ ($k$ will be the number of half-dirty levels)
\FOR {$j=b$ to $u$}
\IF {$\mathcal{D}_j \geq \frac{1}{2} \cdot \frac{\epsilon}{\beta} \cdot \mathcal{C}_j$}
\STATE $k \leftarrow k + 1$  
\STATE $i_{k} \leftarrow j$
\ENDIF
\ENDFOR
\FOR {$q=1$ to $k$}
\STATE $j=i_q-1$
\STATE crit $\leftarrow$ TRUE
\WHILE {$j \geq b$}
\IF {$j$ is upper bound of a dirty range $(j',j)$}
\STATE $\mathcal{D}_{j'+1}^{i_q} \leftarrow \mathcal{D}_{j'+1}^{j} + \mathcal{D}_{j+1}^{i_q}$
\STATE $\mathcal{C}_{j'+1}^{i_q} \leftarrow \mathcal{C}_{j'+1}^{j} + \mathcal{C}_{j+1}^{i_q}$
\STATE Remove dirty range $(j',j)$
\STATE $j \leftarrow j'$
\ELSE
\STATE $\mathcal{D}_j^{i_q} \leftarrow \mathcal{D}_j + \mathcal{D}_{j+1}^{i_q}$
\STATE $\mathcal{C}_j^{i_q} \leftarrow \mathcal{C}_j + \mathcal{C}_{j+1}^{i_q}$
\IF {$\mathcal{D}_j^{i_q} < \frac{1}{2} \cdot \frac{\epsilon}{\beta} \cdot \mathcal{C}_j^{i_q}$}
\STATE Add $(j,i_q)$ to set of dirty ranges
\STATE crit $\leftarrow$ FALSE
\STATE $j \leftarrow b-1$
\ELSE
\STATE $j \leftarrow j-1$
\ENDIF
\ENDIF
\ENDWHILE
\IF {crit}
\STATE HighestCrit $\leftarrow i_q$
\STATE Add $(b-1,i_q)$ to set of dirty ranges
\ENDIF
\ENDFOR
\STATE return HighestCrit
\end{algorithmic}
\end{algorithm}

For the running time of this procedure, creating the sorted list of half-dirty levels takes $O(\log_{\beta} n)$ time. Next, notice that each progression can be done in $O(1)$ time. Since we skip dirty ranges, the only levels that we can check more than once are levels that are lower bounds of a dirty range. In the example above, we check $j_1$ first when we see that $\mathcal{D}_{j_1}^{i_1} < \frac{1}{2} \cdot \frac{\epsilon}{\beta} \cdot \mathcal{C}_{j_1}^{i_1}$, and it is possible we check it for a second time when we check if $\mathcal{D}_{j_1}^{i_2} \geq \frac{1}{2} \cdot \frac{\epsilon}{\beta} \cdot \mathcal{C}_{j_1}^{i_2}$ after skipping the dirty range $(j_1,i_1)$. It is obviously possible that we check it many more times, if each time we ``fail" at level $j_1$. The observation though is that once we check a certain level which was previously checked, then a dirty range is deleted. Since there are up to $u-b = O(\log_{\beta}n)$ dirty ranges, we can check an already checked level only $O(\log_{\beta}n)$ times, which takes $O(\log_{\beta}n)$ time. Other than that, each level is checked only once, and again there are $u-b = O(\log_{\beta}n)$ levels, so this takes $O(\log_{\beta}n)$ time as well. Adding and removing dirty ranges to the list can also be done in $O(1)$ time each. Thus, the overall worst case running time of this procedure is $O(\log_{\beta}n)$. We conclude this paragraph with the following corollary:

\begin{cor}
The worst case running time of finding a half-critical level $i_{crit}$ once the system is dirty is $O(\log_{\beta}n)$.
\end{cor}

\noindent Finally, we address the second question, where our goal is to show that between each update step only one partial reset is necessary. If this is indeed the case, then we can conclude that in the worst case, at each update step we spend $O(\log_{\beta}n)$ time to find a a half-critical level to do a reset up to, which is not the bottleneck of the update time. Since we perform the reset up to the \emph{highest} half-critical level up to $u$, we claim that after the reset there is no half-critical level up to u. Indeed, all $\mathcal{D}_j$ counters for $j \leq i_{crit}$ were reset to $0$, thus following the reset a level $j \leq i_{crit}$ cannot be half-critical. On the other hand, since no $\mathcal{D}_j$ was raised during the reset for \emph{any} $j$, had a critical level existed above $i_{crit}$ and below $u$, it would have existed before the reset as well, a contradiction to the fact that $i_{crit}$ is the highest such level. Thus, following the reset there is no half-critical level up to $u$, which by Claim \ref{negl1} means that the system is not dirty. We conclude this paragraph with the following observation:

\begin{obs} \label{onereset}
Between each two update steps, there cannot be more than one partial reset.
\end{obs}

\noindent \textbf{Computing the sets $\tilde{\mathcal{U}}$ and $\tilde{\mathcal{F}}$:} 
Using the lists $\mathcal{L}_j$, where $j$ ranges over all levels from $0$ to $i_{crit}$, we copy every element $e$ in these lists to the set $\tilde{\mathcal{U}}$, within time that is linear in the number of elements in the resulting $\tilde{\mathcal{U}}$. The set $\tilde{\mathcal{U}}$ too will be maintained as a doubly-linked list. Recall that the set $\tilde{\mathcal{F}}$ consists of all sets $S$ such that $l_{cov}(S) \le i_{crit}$, which contain at least one element from $\tilde{\mathcal{U}}$. To compute $\tilde{\mathcal{F}}$, we scan the list $\tilde{\mathcal{U}}$, and for each $e \in \tilde{\mathcal{U}}$, we scan all sets in $\mathcal{F}_e$ and put in $\tilde{\mathcal{F}}$ the correct ones; 
the time for computing $\tilde{\mathcal{F}}$ is at most $O(|\tilde{\mathcal{U}}| \cdot f)$, since we can naively scan $\mathcal{F}_e$ for all $e \in \tilde{\mathcal{U}}$ (each has at most $f$ frequency) and determine which among them belongs to $\tilde{\mathcal{F}}$. 
We also note that before computing the set $\tilde{\mathcal{F}}$, we should first empty the Cov sets for all sets that currently cover at any level up to $i_{crit}$. Since all these sets are contained in $\tilde{\mathcal{F}}$ that we compute now, this step can be implemented in time $O(|\tilde{\mathcal{F}}|) = O(|\tilde{\mathcal{U}}| \cdot f)$. 
Summarizing, the total runtime for computing the sets $\tilde{\mathcal{U}}$ and $\tilde{\mathcal{F}}$ is $O(|\tilde{\mathcal{U}}| \cdot f)$, thus it exceeds that number of elements that enter a partial reset in $\tilde{\mathcal{U}}$ by a factor of $f$. Since we view all elements in $\tilde{\mathcal{U}}$ as changing their level, 
the amortized update time incurred by the computation of the sets $\tilde{\mathcal{U}}$ and $\tilde{\mathcal{F}}$ throughout all partial resets, as described above, 
does not exceed the total amortized number of element level changes
due to partial resets by more than a factor of $f$,

\bigskip

\noindent \textbf{Applying the static algorithm on $\tilde{\mathcal{U}} \cup \tilde{\mathcal{F}}$:}
For each set $S \in \tilde{\mathcal{F}}$ we keep track of $|S \cap \tilde{\mathcal{U}}|$, where recall that in each iteration of the partial reset we remove from $\tilde{\mathcal{U}}$ the elements that become covered, until $\tilde{\mathcal{U}} = \emptyset$. We maintain another ``level system", where each set $S \in \tilde{\mathcal{F}}$ is placed in a level $j$ such that $j = \lfloor \log_{\beta} (\frac{|S \cap \tilde{\mathcal{U}}|}{c(S)}) \rfloor$. To initialize this level system, for each $e \in \tilde{\mathcal{U}}$ we go over all sets $S \in \mathcal{F}_e$ and increment their $|S \cap \tilde{\mathcal{U}}|$ value by one. Once we have all of these values, we place each set in its level in $O(1)$ time. Initializing this system takes $O(|\tilde{\mathcal{U}}| \cdot f)$ time, since $|\tilde{\mathcal{F}}| = O(|\tilde{\mathcal{U}}| \cdot f)$. In addition, we maintain a pointer to the highest occupied level. At each iteration, we choose an arbitrary set $S$ from the highest occupied level. $S$ then joins the set cover, we let $Cov(S) = S \cap \tilde{\mathcal{U}}$, and we remove $S \cap \tilde{\mathcal{U}}$ from $\tilde{\mathcal{U}}$. Then, we must recompute the values $|S' \cap \tilde{\mathcal{U}}|$ for all other sets $S' \in \tilde{\mathcal{F}}$. To do so we simply consider each element $e \in S \cap \tilde{\mathcal{U}}$, and for each $S' \in \mathcal{F}_e$ we reduce the $|S' \cap \tilde{\mathcal{U}}|$ value by one, and check if $S'$ needs to move down a level (all in $O(1)$ time). Since we take $O(f)$ time for each element removal from $\tilde{\mathcal{U}}$, and each element is removed only once, the total runtime of this process is $O(|\tilde{\mathcal{U}}| \cdot f)$. 

The only issue remaining is the time spent on maintaining a pointer to the highest occupied level. Notice that the pointer cannot rise as we move forward with the iterations. If the pointer points to a now unoccupied level, it must search for the highest occupied level by descending one level at a time. Thus, the total runtime of maintaining this pointer is $O(|\tilde{\mathcal{F}}| + L) = O(|\tilde{\mathcal{U}}| \cdot f + L)$, where $L$ is the number of levels in this system that the pointer will point to.

\begin{lem} \label{howmanylevels}
$L=O(\log_{\beta}n)$.
\end{lem}

\begin{proof}
Denote by $u'$ and $b'$ the highest and lowest levels that the pointer will point to throughout the reset, respectively. This means that the first Cov set will be created at level $u'$, and the last at level $b'$. By Observation \ref{ob:added}:

\begin{equation} \label{0.25par}
u' \leq i_{crit}+1
\end{equation} 

\noindent Denote by $S'$ the last chosen set in the reset procedure, taken to level $b'$. Thus: 

\begin{equation} \label{0.5par}
b' = \Biggl\lfloor \log_{\beta} \left(\frac{|S' \cap \tilde{\mathcal{U}}|}{c(S')}\right) \Biggr\rfloor > \log_{\beta} \left(\frac{1}{c(S')}\right) - 1 = \log_{\beta} \left(\frac{1}{\beta \cdot c(S')}\right)
\end{equation}

\noindent $S'$ contains at least one element, $e'$, which was in $\tilde{\mathcal{U}}$ when $S'$ was taken to the set cover. Denote by $S''$ the set which covered $e'$ right before the reset (where perhaps $S' = S''$). $S'' \in \tilde{\mathcal{F}}$, since $e' \in \tilde{\mathcal{U}}$. If $S' \neq S''$, then we know that $S''$ was not taken to the set cover. Indeed, both $S'$ and $S''$ contain $e'$, and since $S'$ is the last taken set to the set cover, had we taken $S''$ to the set cover before, $e'$ would have already been removed from $\tilde{\mathcal{U}}$ by the time we reach the last iteration. In that case, right before the last iteration where $S'$ was taken to the set cover, we had:

$$\frac{|S' \cap \tilde{\mathcal{U}}|}{c(S')} \geq \frac{|S'' \cap \tilde{\mathcal{U}}|}{c(S'')}.$$

\noindent Notice that if $S'=S''$ this also holds. Now:

$$\frac{n}{c(S')} \geq \frac{|S' \cap \tilde{\mathcal{U}}|}{c(S')} \geq \frac{|S'' \cap \tilde{\mathcal{U}}|}{c(S'')} \geq \frac{1}{c(S'')},$$

\noindent where the last inequality holds since $e' \in S'' \cap \tilde{\mathcal{U}}$ right before the last iteration. Thus:

\begin{equation} \label{1par}
c(S'') \geq \frac{c(S')}{n}
\end{equation}

\noindent Since $S''$ was in the set cover before the reset, let $l = l_{cov}(S'')$ where $l \geq b$. Recall that $b$ is the lowest level of a set in the set cover right before the reset. When $Cov(S'')$ was created, either by local rise or partial reset, it satisfied $|Cov(S'')| \geq c(S'') \cdot \beta^l$, by Observations \ref{obclean} and \ref{parclean} respectively. Denote by $x''$ the size of $Cov(S'')$ when it was created. So:

$$n \geq x'' \geq c(S'') \cdot \beta^l \geq c(S'') \cdot \beta^b,$$

\noindent and we obtain:

\begin{equation} \label{2par}
c(S'') \leq \frac{n}{\beta^b}
\end{equation}

\noindent Combining Equations (\ref{1par}) and (\ref{2par}), and rearranging, we get:

\begin{equation} \label{3par}
b \leq \log_{\beta}\left(\frac{n^2}{c(S')}\right)
\end{equation}

\noindent Now, we can take Equations (\ref{0.5par}) and (\ref{3par}) and write:

\begin{equation} \label{4par}
b-b' < \log_{\beta}\left(\frac{n^2}{c(S')}\right) - \log_{\beta} \left(\frac{1}{\beta \cdot c(S')}\right) = 2\log_{\beta}(n) + 1
\end{equation}

\noindent Moreover, we know that $i_{crit} \leq u$, where recall that $u = b + 4 \log_{\beta}n$. Thus, by Equation (\ref{0.25par}), we know that $u' \leq u+1 = b + 4 \log_{\beta}n + 1$. Since $L = u' - b' + 1$, we can write:

$$L = u' - b' + 1 = u' - b + b - b' + 1 < 4 \log_{\beta}n + 1 + 2\log_{\beta}(n) + 1 + 1 = 6\log_{\beta}(n) + 3$$
\qed
\end{proof}

\noindent We conclude that once we find $i_{crit}$ and the sets $\tilde{\mathcal{U}}$ and $\tilde{\mathcal{F}}$, the running time of a partial reset is $O(|\tilde{\mathcal{U}}| \cdot f + \log_{\beta}n)$. Combining this with the time needed to find $i_{crit}$ and the sets $\tilde{\mathcal{U}}$ and $\tilde{\mathcal{F}}$, we derive the following corollary.

\begin{cor} \label{rstime}
The running time of a partial reset (including finding $i_{crit}$ and the sets $\tilde{\mathcal{U}}$ and $\tilde{\mathcal{F}}$) is $O(|\tilde{\mathcal{U}}| \cdot f + \log_{\beta}n)$.
\end{cor}

\noindent Recall that we aim to show that the amortized running time is $O_{\epsilon}(f \cdot \ln n)$. By Observation \ref{onereset}, up to only one reset can occur every update step. Thus, the $O(\log_{\beta}n)$ addend of the running time of the partial reset is not the bottleneck, meaning we can afford an extra $O(\log_{\beta}n)$ time each update step and that would not change the asymptotic amortized update time. Since we view each element in $\tilde{\mathcal{U}}$ as changing its level in the reset, and we spend $O(f)$ time for each element level change, we conclude that the $O(|\tilde{\mathcal{U}}| \cdot f)$ addend is not the bottleneck either. Thus, the time spent by implementing the partial resets throughout the update sequence is dominated by the number of element level changes throughout the sequence multiplied by $O(f)$. 

\bigskip

\noindent We conclude this section with the following corollary:
\begin{cor} \label{updttime}
The amortized update time of our algorithm is $O(\eps^{-5} \cdot f \cdot \ln n) = O_\eps(f \cdot \ln n)$.
\end{cor}

\section{Recourse Bound} \label{rcrs}

\noindent In this section we prove that the amortized recourse is $O_{\epsilon}(\log C)$, where we assume that $C<n$, otherwise we can use Observation \ref{clm3} to bound the amortized recourse by $O_{\epsilon}(\ln n)$. Consider a local rise of some $Cov(S)$ to level $j$. The elements in this newly created $Cov(S)$ could have arrived from any level $i$ such that $i<j$. Denote by $Cov_i^j(S)$ the set of elements that joined this local rise of $Cov(S)$ (to level $j$) that were covered at level $i$ (where $i<j$) right before the local rise. Let $Z_i^j(S) = |Cov_i^j(S)|$. Notice that $\sum_{i=0}^{j-1} Z_i^j(S) = |Cov(S)|$. We begin with the following important claim.

\begin{cl} \label{second}
$Z_i^j(S) \leq \lceil c(S) \cdot \beta^{i+2}\rceil$
\end{cl}

\begin{proof}
Since $Cov_i^j(S) \subseteq N_{i+1}(S)$, we have $Z_i^j(S) \leq |N_{i+1}(S)|$, thus it suffices to prove that $|N_{i+1}(S)| \leq \lceil c(S) \cdot \beta^{i+2}\rceil$. Assume by contradiction that $|N_{i+1}(S)| > \lceil c(S) \cdot \beta^{i+2}\rceil$ right after time step $t$. We split to two cases. In the first case, right before time step $t$ we had $|N_{i+1}(S)| = \lceil c(S) \cdot \beta^{i+2}\rceil$, and then $|N_{i+1}(S)|$ grew by one (or more) at $t$ to exceed $\lceil c(S) \cdot \beta^{i+2}\rceil$. In the second case, right before time step $t$ we had $|N_{i+1}(S)| < \lceil c(S) \cdot \beta^{i+2}\rceil$, and then $|N_{i+1}(S)|$ grew by more than one at $t$ to exceed $\lceil c(S) \cdot \beta^{i+2}\rceil$. 

In the first case, when $|N_{i+1}(S)| = \lceil c(S) \cdot \beta^{i+2}\rceil$, if $l_{cov}(S) \leq i$ then $S$ is $(i+1)$-PD, thus a local rise would immediately occur (to level at least $i+1$), and consequently $|N_{i+1}(S)| = 0$ (since $N_{i+1}(S)$ would be taken as $Cov(S)$ at level $i+1$), a contradiction. If $l_{cov}(S) > i$, then by Invariant \ref{inv3}, $|N_{i+1}(S)| = 0$, a contradiction. In the second case, $|N_{i+1}(S)|$ can grow by more than one instantaneously only due to a partial reset (insertion can raise it by only one). By Claim \ref{clm2}, $|N_{i+1}(S)|$ cannot grow as a result of a partial reset with a half-critical level $i_{crit} \leq i$. Thus, the half-critical level of the partial reset must satisfy $i_{crit} \geq i+1$. If $S$ participated in this partial reset, then by Claim \ref{clm1}, following the reset $|N_{i+1}(S)| < c(S) \cdot \beta^{i+1}$, a contradiction. If $S$ did not participate in the reset, then $l_{cov}(S) > i_{crit}$, since $S$ contains elements that were in the set $\tilde{\mathcal{U}}$ of the reset (all elements that following the reset are in $N_{i+1}(S)$). But $S$ cannot have elements which it can cover in the set $\tilde{\mathcal{U}}$ in the reset, because that would mean they are covered at a level $\leq i_{crit}$ right before the reset, when $l_{cov}(S)>i_{crit}$, a contradiction to Invariant \ref{inv3}. 
\qed
\end{proof}

\begin{cl} \label{fourth}
The number of element level changes due to local rises from level $i$ to any level greater than $i$, throughout the whole sequence, is $O(\epsilon^{-2} \cdot K)$, where $K$ is the length of the update sequence.
\end{cl}

\begin{proof}
\noindent We denote by $L_j$ the number of element level changes in the whole sequence due to local rises to level $j$, and by $a_j$ the number of such local rises. Denote by $S_1,S_2,\ldots,S_{a_j}$ the sets that perform local rises to level $j$ (in chronological order), where possibly $S_{i_1}$ and $S_{i_2}$ are the same set (for any $i_1$ and $i_2$). Therefore, by Observation \ref{obclean} we know that $L_j \geq \sum_{p=1}^{a_j}\beta^j \cdot c(S_p)$. For any $i<j$: 

\begin{align*}
L_j & \geq \beta^j \cdot \sum_{p=1}^{a_j}c(S_p) \\ & \geq 2 \cdot \lceil \beta^i \rceil \cdot \frac{1}{4} \cdot \beta^{j-i} \cdot \sum_{p=1}^{a_j}c(S_p) \\ & = \frac{1}{4} \cdot \beta^{j-i} \cdot \sum_{p=1}^{a_j}2 \cdot \lceil \beta^i \rceil \cdot c(S_p) \\ & \geq \frac{1}{4} \cdot \beta^{j-i} \cdot \sum_{p=1}^{a_j} \lceil \beta^{i+2} \rceil \cdot c(S_p) \\ & \geq \frac{1}{4} \cdot \beta^{j-i} \cdot \sum_{p=1}^{a_j} Z_i^j(S_p),
\end{align*}

\noindent where the second transition is due to $\lceil \beta^i \rceil \geq 2$, the fourth for $\epsilon \leq \sqrt{2}-1$, and the last due to Claim \ref{second}. Denoting by $Z_i^j$ the total number of level changes of elements that joined any local rise to level $j$ from level $i$ (i.e., $Z_i^j = \sum_{p=1}^{a_j} Z_i^j(S_p)$), we conclude that:

\begin{equation} \label{eqfirst}
Z_i^j \leq 4 \cdot L_j \cdot \frac{1}{\beta^{j-i}}~.
\end{equation}

\noindent By Lemma \ref{first},   $L_j = O(\epsilon^{-1} \cdot K)$. Let $Z_i = \sum_{j=i+1}^{\infty} Z_i^j$, i.e., $Z_i$ is the number of element level changes due to local rises from level $i$ to any level greater than $i$, throughout the whole sequence. So:

\begin{equation*} 
Z_i  ~=~ \sum_{j=i+1}^{\infty} Z_i^j ~\leq~ \sum_{j=i+1}^{\infty} 4 \cdot L_j \cdot \frac{1}{\beta^{j-i}} ~=~ O(\epsilon^{-2} \cdot K)
\end{equation*}
\qed
\end{proof}

\noindent Observe that the number of element level changes due to insertions/deletions throughout the whole sequence is $O(K)$, since only the inserted/deleted element ``changes" its level.
Combining this observation
with Claim \ref{fourth}, we derive the following corollary.  

\begin{cor} \label{sixth}
The total number of element level changes due to insertions, deletions and local rises from level $i$ to any level, throughout the whole sequence, is $O(\epsilon^{-2} \cdot K)$. 
\end{cor}

\noindent Let $q$ denote the number of partial resets throughout the entire sequence, let $i_p$ denote the half-critical level of the $p$th reset, and denote by $S_p \subseteq \mathcal{F}$ the collection of sets \textbf{at level $\geq 0$} (meaning in the SC) that participate in the $p$th reset.

\begin{lem} \label{seventh}
$\sum_{p=1}^q |S_p| = O(\epsilon^{-4} \cdot \log C \cdot K)$.
\end{lem}

\begin{proof}
Consider the following token scheme where in each update step we give $O(\epsilon^{-4} \cdot \log C)$ tokens, and every time a covering set participates in a reset it pays a token. Thus, proving that we do not run out of tokens would conclude the proof. 

First, let us fix an arbitrary level $i$ in the range $[0,\ceil{\log_\beta C}]$.
Say every time an element changes its level from $i$ to any other level due to insertion, deletion or local rise, it must pay 
$2 \cdot \frac{\beta^3}{\epsilon}$ tokens to its previous covering level, $i$.
By Corollary \ref{sixth}, the total number of element level changes due to insertions, deletions and local rises from level $i$ to any level is $c_{\epsilon} \cdot K$, where $c_{\epsilon} = O(\epsilon^{-2})$. Hence, the total number of tokens paid in this way to Cov sets of elements is $O(\frac{\beta^3}{\epsilon} \cdot  \epsilon^{-2} \cdot K) = O(\eps^{-3} \cdot K)$. 
Over all levels in the range $[0,\ceil{\log_\beta C}]$, 
the total number of tokens paid to Cov sets is $O(\epsilon^{-3} \cdot K \cdot \log_\beta C) = O(\epsilon^{-4} \cdot \log C  \cdot K)$.

Next, let us fix an arbitrary level $i$ in the complementary range $[\ceil{\log_\beta C}+1, \ceil{\log_\beta (n \cdot C)}]$, where notice that $\ceil{\log_\beta (n \cdot C)}$ is the highest possible level.
Write $i = \ceil{\log_\beta C} + i'$, where $1 \le i' \le \ceil{\log_\beta (n \cdot C)} - \ceil{\log_\beta C}$.
Say every time an element changes its level from $i = \ceil{\log_\beta C} + i'$ to any other level due to insertion, deletion or local rise, it must pay 
$2 \cdot \frac{\beta^3}{\epsilon} \cdot \frac{1}{\beta^{i'}}$ tokens to its previous covering level, $i$. By Corollary \ref{sixth}, the total number of element level changes due to insertions, deletions and local rises from level $i$ to any level is $c_{\epsilon} \cdot K$, where $c_{\epsilon} = O(\epsilon^{-2})$. Hence, the total number of tokens paid in this way to Cov sets of elements is 
$O(\frac{\beta^3}{\epsilon} \cdot \frac{1}{\beta^{i'}} \cdot  \epsilon^{-2}  \cdot K)
= O(\epsilon^{-3} \cdot \frac{K}{\beta^{i'}})$. Over all levels in
the range $[\ceil{\log_\beta C}+1, \ceil{\log_\beta (n \cdot C)}]$, the total number of tokens paid to Cov sets is 
$$\sum_{i = \ceil{\log_\beta C}+1}^{\ceil{\log_\beta (n \cdot C)}} O\left(\epsilon^{-3} \cdot \frac{K}{\beta^{i'}}\right) 
~=~ \sum_{i' = 1}^{\ceil{\log_\beta (n \cdot C)} - \ceil{\log_\beta C}} O\left(\epsilon^{-3} \cdot \frac{K}{\beta^{i'}}\right) 
~<~ \sum_{i' = 1}^{\infty} O\left(\epsilon^{-3} \cdot \frac{K}{\beta^{i'}}\right) 
~=~ O(\epsilon^{-4} \cdot K).$$ 

Next, we show that these tokens can be redistributed so that every set $S$ that participates in a partial reset as a covering set gets a fresh token, which would complete the proof. Recall that once an element $e$ leaves its covering set $S'$ at a level $i$, we increment $\mathcal{D}_i^i$ by $\frac{1}{\beta^i}$. Since $i_{crit}$ is half-critical, for every $j \leq i_{crit}$ we have $\mathcal{D}_j^{i_{crit}} \geq \frac{1}{2} \cdot \frac{\epsilon}{\beta} \cdot \mathcal{C}_j^{i_{crit}}$. We argue inductively on $j$, for any $j = i_{crit}, i_{crit} - 1, \ldots, 0$ that we can redistribute one token for each set covering at a level between $j$ and $i_{crit}$ right before the reset, by using \emph{only} tokens obtained from elements that left levels between $j$ and $i_{crit}$ that were necessary to raise $\mathcal{D}_{j}^{i_{crit}}$ to \emph{exactly} $\frac{1}{2} \cdot \frac{\epsilon}{\beta} \cdot \mathcal{C}_{j}^{i_{crit}}$. 

For the basis $j=i_{crit}$, since level $i_{crit}$ is half-critical (and thus half-dirty as well), $\mathcal{D}_{i_{crit}}^{i_{crit}} \geq \frac{1}{2} \cdot \frac{\epsilon}{\beta} \cdot \mathcal{C}_{i_{crit}}^{i_{crit}}$.
Consider first level $i_{crit}$ in the range $[0,\ceil{\log_\beta C}]$. Since each leaving element from level $i_{crit}$ gives $\frac{2\beta^3}{\epsilon}$ tokens, level $i_{crit}$ has obtained $\mathcal{D}_{i_{crit}}^{i_{crit}} \cdot \beta^{i_{crit}} \cdot \frac{2\beta^3}{\epsilon}$ tokens. And since level $i_{crit}$ is half-dirty, it has obtained $\geq \frac{1}{2} \cdot \frac{\epsilon}{\beta} \cdot \mathcal{C}_{i_{crit}}^{i_{crit}} \cdot \beta^{i_{crit}} \cdot \frac{2\beta^3}{\epsilon} = \mathcal{C}_{i_{crit}}^{i_{crit}} \cdot \beta^{i_{crit}+2}$ tokens. Notice that for any covering set $S$ such that $l_{cov}(S) = i_{crit}$, we know that $c(S) \cdot \beta^{i_{crit}+2} > 1$, otherwise the set $S$ would be PD when created. Summing on all sets which cover at level $i_{crit}$, we get that $\mathcal{C}_{i_{crit}}^{i_{crit}} \cdot \beta^{i_{crit}+2} > |\mathcal{S}_{i_{crit}}^{i_{crit}}|$. Thus, we have more than $|\mathcal{S}_{i_{crit}}^{i_{crit}}|$ tokens by using the tokens obtained by \emph{only} $\frac{1}{2} \cdot \frac{\epsilon}{\beta} \cdot \mathcal{C}_{i_{crit}}^{i_{crit}} \cdot \beta^{i_{crit}}$ elements leaving the level, and we save the rest. Now consider level $i_{crit} > \ceil{\log_\beta C}$, and write $i_{crit} = \ceil{\log_\beta C} + i'$. Since each leaving element from level $i_{crit}$ gives $\frac{2\beta^3}{\epsilon} \cdot \frac{1}{\beta^{i'}}$ tokens, level $i_{crit}$ has obtained $\mathcal{D}_{i_{crit}}^{i_{crit}} \cdot \beta^{i_{crit}} \cdot \frac{2\beta^3}{\epsilon} \cdot \frac{1}{\beta^{i'}}$ tokens. And since level $i_{crit}$ is half-dirty, it has obtained $\geq \frac{1}{2} \cdot \frac{\epsilon}{\beta} \cdot \mathcal{C}_{i_{crit}}^{i_{crit}} \cdot \beta^{i_{crit}} \cdot \frac{2\beta^3}{\epsilon} \cdot \frac{1}{\beta^{i'}} = \mathcal{C}_{i_{crit}}^{i_{crit}} \cdot \beta^{i_{crit} - i' +2}$ tokens. Consider a set $S$ such that $l_{cov}(S)=i_{crit}$. By Observations \ref{obclean} and \ref{parclean}, we know that when $Cov(S)$ was created, we had:

$$|Cov(S)| \geq c(S) \cdot \beta^{i_{crit}} = c(S) \cdot \beta^{\ceil{\log_\beta C} + i'} \geq c(S) \cdot C \cdot \beta^{i'} \geq \beta^{i'}$$

\noindent We also know that $|Cov(S)| < c(S) \cdot \beta^{i_{crit}+2}$ (no PD sets). Thus, we get that $c(S) \cdot \beta^{i_{crit}+2} > \beta^{i'}$, meaning $c(S) \cdot \beta^{i_{crit}-i'+2} > 1$. Summing up on all $S \in \mathcal{S}_{i_{crit}}^{i_{crit}}$, we get that $\mathcal{C}_{i_{crit}}^{i_{crit}} \cdot \beta^{i_{crit}-i'+2} > |\mathcal{S}_{i_{crit}}^{i_{crit}}|$. Therefore, we have more than $|\mathcal{S}_{i_{crit}}^{i_{crit}}|$ tokens by using the tokens obtained by \emph{only} $\frac{1}{2} \cdot \frac{\epsilon}{\beta} \cdot \mathcal{C}_{i_{crit}}^{i_{crit}} \cdot \beta^{i_{crit}}$ elements leaving the level, and we save the rest.

For the induction step, assume the claim holds inductively for some $1 \leq j \leq i_{crit}$, and  prove it for $j-1$. Since level $i_{crit}$ is half-critical, $\mathcal{D}_{j}^{i_{crit}} \geq \frac{1}{2} \cdot \frac{\epsilon}{\beta} \cdot \mathcal{C}_{j}^{i_{crit}}$. Assume that $\mathcal{D}_j^{i_{crit}} = \frac{1}{2} \cdot \frac{\epsilon}{\beta} \cdot \mathcal{C}_j^{i_{crit}} + X$, for some $X \geq 0$. By the induction hypothesis, we have not yet used the tokens from leaving elements that caused the overhead spare of $X$ for $\mathcal{D}_j^{i_{crit}}$. These elements left from levels between $j$ and $i_{crit}$, and we want to give a lower bound to the number of tokens these leaving elements provide.  To lower bound the number of tokens the ``extra" leaving elements provide, given that they contribute $X$ to $\mathcal{D}_j^{i_{crit}}$, we must find the level between $j$ and $i_{crit}$ that \emph{minimizes} the ratio of tokens per dirt unit. Notice that for each level $l > \ceil{\log_\beta C}$, the ratio of $l$ denoted by $r_l$ is:

$$r_l = \frac{2 \cdot \beta^3}{\epsilon \cdot \beta^{l - \ceil{\log_\beta C}}} \cdot \frac{\beta^l}{1} = \frac{2 \cdot \beta^{\ceil{\log_\beta C} + 3}}{\epsilon},$$

\noindent meaning the ratio is independent of the level (if it is above $\ceil{\log_\beta C}$). For each level $l$ in the range $[0,\ceil{\log_\beta C}]$, the ratio is:

$$r_l = \frac{2 \cdot \beta^3}{\epsilon} \cdot \frac{\beta^l}{1} = \frac{2 \cdot \beta^{l + 3}}{\epsilon},$$  

\noindent meaning the smaller the level the smaller the ratio. We split to two cases, where in the first $j-1$ is in the range $[0,\ceil{\log_\beta C}]$, and in the second $j-1 > \ceil{\log_\beta C}$. In the first case, level $j$ minimizes the ratio. It could be though that $j-1 = \ceil{\log_\beta C}$, in that case we know that the ratio for $j$, denoted $r_j$, satisfies:

$$ r_j \geq \frac{2 \cdot \beta^{j + 2}}{\epsilon}.$$

In the second case, since $j-1 > \ceil{\log_\beta C}$, then any level between $j$ and $i_{crit}$ is also $> \ceil{\log_\beta C}$. Therefore the ratio of any level higher than $j-1$ is the same, $\frac{2 \cdot \beta^{\ceil{\log_\beta C} + 3}}{\epsilon}$. Multiplying the ratios by $X$, we get that the minimum number of obtained tokens from the ``extra" leaving elements for the first and second case is $X \cdot \frac{2 \cdot \beta^{j + 2}}{\epsilon}$ and $X \cdot \frac{2 \cdot \beta^{\ceil{\log_\beta C} + 3}}{\epsilon}$ respectively.

In addition to the tokens not used yet from the overhead of $X$, the tokens obtained from elements leaving level $j-1$ have not been used yet. Since $\mathcal{D}_{j-1}^{i_{crit}} \geq \frac{1}{2} \cdot \frac{\epsilon}{\beta} \cdot \mathcal{C}_{j-1}^{i_{crit}}$, we can write:

$$\mathcal{D}_{j}^{i_{crit}} + \mathcal{D}_{j-1}^{j-1} = \frac{1}{2} \cdot \frac{\epsilon}{\beta} \cdot (\mathcal{C}_{j}^{i_{crit}} + \mathcal{C}_{j-1}^{j-1}) + Y,$$

\noindent where $Y \geq 0$, and we want to prove that the tokens obtained from this overhead of $Y$ are not needed for level $j-1$ (and we can save them for later). Since $\mathcal{D}_j^{i_{crit}} = \frac{1}{2} \cdot \frac{\epsilon}{\beta} \cdot \mathcal{C}_j^{i_{crit}} + X$, we get that:

$$\frac{1}{2} \cdot \frac{\epsilon}{\beta} \cdot \mathcal{C}_j^{i_{crit}} + X + \mathcal{D}_{j-1}^{j-1} = \frac{1}{2} \cdot \frac{\epsilon}{\beta} \cdot (\mathcal{C}_{j}^{i_{crit}} + \mathcal{C}_{j-1}^{j-1}) + Y,$$

\noindent and thus:

\begin{equation} \label{dj}
\mathcal{D}_{j-1}^{j-1} = \frac{1}{2} \cdot \frac{\epsilon}{\beta} \cdot \mathcal{C}_{j-1}^{j-1} - X + Y
\end{equation}

In the first case, where $j-1$ is in the range $[0,\ceil{\log_\beta C}]$, level $j-1$ has obtained $\mathcal{D}_{j-1}^{j-1} \cdot \frac{2\beta^{j+2}}{\epsilon}$ tokens. Thus, by Equation (\ref{dj}), it obtained $(\frac{1}{2} \cdot \frac{\epsilon}{\beta} \cdot \mathcal{C}_{j-1}^{j-1} - X + Y) \cdot \frac{2\beta^{j+2}}{\epsilon}$ tokens, meaning $\beta^{j+1} \cdot \mathcal{C}_{j-1}^{j-1} - \frac{2\beta^{j+2}}{\epsilon} \cdot X + \frac{2\beta^{j+2}}{\epsilon} \cdot Y$ tokens. Since the overhead of $X$ gave us $\geq \frac{2\beta^{j+2}}{\epsilon} \cdot X$ tokens, we have a total of $\geq \beta^{j+1} \cdot \mathcal{C}_{j-1}^{j-1} + \frac{2\beta^{j+2}}{\epsilon} \cdot Y$ tokens. Notice that for any covering set $S$ such that $l_{cov}(S) = j-1$, we know that $c(S) \cdot \beta^{j+1} > 1$, otherwise the set $S$ would be PD when created. Summing on all sets which cover at level $j-1$, we get that $\mathcal{C}_{j-1}^{j-1} \cdot \beta^{j+1} > |\mathcal{S}_{j-1}^{j-1}|$. Thus, we have more than $|\mathcal{S}_{j-1}^{j-1}|$ tokens from which we can redistribute one token for each set covering at level $j-1$ right before the reset, \emph{without} using the tokens from leaving elements which caused the overhead of $Y$, and so the induction step holds for the first case.

In the second case, where $j-1 > \ceil{\log_\beta C}$, level $j-1$ has obtained $\mathcal{D}_{j-1}^{j-1} \cdot \frac{2 \cdot \beta^{\ceil{\log_\beta C} + 3}}{\epsilon}$ tokens. Thus, by Equation (\ref{dj}), it obtained $(\frac{1}{2} \cdot \frac{\epsilon}{\beta} \cdot \mathcal{C}_{j-1}^{j-1} - X + Y) \cdot \frac{2 \cdot \beta^{\ceil{\log_\beta C} + 3}}{\epsilon}$ tokens, meaning $\beta^{\ceil{\log_\beta C} + 2} \cdot \mathcal{C}_{j-1}^{j-1} - \frac{2 \cdot \beta^{\ceil{\log_\beta C} + 3}}{\epsilon} \cdot X + \frac{2 \cdot \beta^{\ceil{\log_\beta C} + 3}}{\epsilon} \cdot Y$ tokens. Since the overhead of $X$ gave us $\geq \frac{2 \cdot \beta^{\ceil{\log_\beta C} + 3}}{\epsilon} \cdot X$ tokens, we have a total of $\beta^{\ceil{\log_\beta C} + 2} \cdot \mathcal{C}_{j-1}^{j-1} + \frac{2 \cdot \beta^{\ceil{\log_\beta C} + 3}}{\epsilon} \cdot Y$ tokens. Now, write $j-1 = \ceil{\log_\beta C} + i'$, and consider a set $S$ such that $l_{cov}(S)=j-1$. By Observations \ref{obclean} and \ref{parclean}, we know that when $Cov(S)$ was created, we had:

$$|Cov(S)| \geq c(S) \cdot \beta^{j-1} = c(S) \cdot \beta^{\ceil{\log_\beta C} + i'} \geq c(S) \cdot C \cdot \beta^{i'} \geq \beta^{i'}$$

\noindent We also know that $|Cov(S)| < c(S) \cdot \beta^{j+1}$ (no PD sets). Thus, we get that $c(S) \cdot \beta^{j+1} > \beta^{i'}$, meaning $c(S) \cdot \beta^{\ceil{\log_\beta C}+2} > 1$. Summing up on all $S \in \mathcal{S}_{j-1}^{j-1}$, we get that $\mathcal{C}_{j-1}^{j-1} \cdot \beta^{\ceil{\log_\beta C}+2} > |\mathcal{S}_{j-1}^{j-1}|$. Thus, we have more than $|\mathcal{S}_{j-1}^{j-1}|$ tokens from which we can redistribute one token for each set covering at level $j-1$ right before the reset, \emph{without} using the tokens from leaving elements which caused the overhead of $Y$, and so the induction step holds for the second case as well.

Applying the induction claim in the particular case $j = 0$, we conclude that we can pay all sets covering up to level $i_{crit}$ (one token each) by using tokens from elements that left levels between $0$ and $i_{crit}$ that were necessary to raise $\mathcal{D}_{0}^{i_{crit}}$ to $\frac{1}{2} \cdot \frac{\epsilon}{\beta} \cdot \mathcal{C}_{0}^{i_{crit}}$.
By Equation (\ref{eq1half}), since level $i_{crit}$ is half-critical, $\mathcal{D}_{0}^{i_{crit}} \geq \frac{1}{2} \cdot \frac{\epsilon}{\beta} \cdot \mathcal{C}_{0}^{i_{crit}}$, thus we have enough tokens, and the lemma holds. 
\qed
\end{proof}

\noindent Let $t(S)$ be the total number of times that the set $S$ joins or leaves the SC. Thus, the total number of changes to the SC throughout the entire update sequence, i.e., the total recourse, is $\sum_{S \in \mathcal{F}} t(S)$. Let $r(S)$ be the number of partial resets in which $S$ participated in as a covering set (in $\tilde{\mathcal{F}}$ and $0 \leq l_{cov}(S)$), meaning $S$ was part of the SC right before the partial reset, and it participated in it. Denote by $\mathcal{F}'$ the collection of sets that enter the set cover only once, and never participate in a reset as a covering set. Meaning, every $S' \in \mathcal{F}'$ sometime throughout the sequence enters the set cover, and $r(S') = 0$. Let $m' = |\mathcal{F}'|$.

\begin{obs} \label{eighth}
For any $S \in \mathcal{F}$, $t(S) \leq 2r(S) + 1$.
\end{obs}

\begin{proof}
A covering set may leave the SC only if it participates in a partial reset as a covering set. Clearly, for a set to join the SC it must have previously left it. Thus the number of times a set joins or leaves the SC is $\leq$ twice the number of times it participates in a reset as a covering set, plus $1$ since initially the SC is empty (therefore it began not in the SC). 
\qed
\end{proof}

\begin{obs} \label{eighthh}
The total recourse is $\sum_{S \in \mathcal{F}} t(S) \leq \sum_{S \in \mathcal{F}} (3 \cdot r(S)) + m'$.
\end{obs}

\begin{proof}
For any set $S \notin \mathcal{F'}$, we have that $t(S) \leq 3r(S)$, because either $r(S) \geq 1$ and then by Observation \ref{eighth} $t(S) \leq 3r(S)$, or $S$ never enters the set cover and then $t(S) = 0$. For any set $S' \in \mathcal{F}'$, we have that $t(S') = 1$. Summing up, we get $\sum_{S \in \mathcal{F}} t(S) \leq \sum_{S \in \mathcal{F}} (3 \cdot r(S)) + m'$.
\qed
\end{proof}

\noindent Our goal now is to prove that $m' = O(K)$, where $K$ is the length of the sequence. We partition $\mathcal{F}'$ into three collections of sets, denoted by $\mathcal{F}'_1$, $\mathcal{F}'_2$ and $\mathcal{F}'_3$. $\mathcal{F}'_1$ contains all the sets $S'_1 \in\mathcal{F'}$ such that the one time that $S'_1$ enters the set cover, it is via insertion. $\mathcal{F}'_2$ contains all the sets $S'_2 \in\mathcal{F'}$ such that the one time that $S'_2$ enters the set cover, it is via local rise. $\mathcal{F}'_3$ contains all the sets $S'_3 \in\mathcal{F'}$ such that the one time that $S'_3$ enters the set cover, it is via partial reset. Let $m'_i = |\mathcal{F}'_i|$ for $i=1,2,3$.

\begin{cl} \label{insK}
$m'_1 \leq K$.
\end{cl}

\begin{proof}
Since we only have up to one insertion per update step, the number of sets that joined the set cover due to insertion, meaning when element $e$ is inserted and all sets in $\mathcal{F}_e$ are not in the SC thus one must join it, is clearly $\leq K$. Therefore, $m'_1 \leq K$.
\qed
\end{proof}

\begin{cl} \label{lrK}
$m'_2 \leq K$.
\end{cl}

\begin{proof}
We will prove that in each update step there could be only up to one local rise. If this is indeed the case, then clearly the number of sets that joined the set cover due to local rise is $\leq K$, and so we would get that $m'_2 \leq K$. Assume by contradiction that following the update step at time $t'$ there are two local rises, of sets $S_1$ and $S_2$ which are $j_1$-PD and $j_2$-PD, respectively. Without loss of generality, assume that the local rise of $S_1$ occurs first. Recall that if we have multiple PD sets, or a set which is PD with respect to multiple levels, we perform the local rise to the \emph{highest} level. Moreover, a local rise can only raise the level of elements, thus a local rise cannot trigger another one. If $S_1 = S_2$, then following the first local rise we have $|N_{j_1+1}(S_1)| = 0$. Now, if $S_1$ is $j_2$-PD, this means that $j_2$ is higher than $j_1+1$, because $|N_{j}(S_1)| = 0$ for any $j \leq j_1+1$. But that would mean that $S_1$ was $j_2$-PD before the first local rise as well, since the first local rise can only raise the level of elements. This is a contradiction to the fact that we perform a local rise to the highest PD level. If $S_1 \neq S_2$, then the inserted element, $e$, which triggers the local rises, is clearly contained in $S_1$ and $S_2$, and in $N_{j_1}(S_1)$ and $N_{j_2}(S_2)$. Since the local rise of $S_1$ occurs first, we know that $j_1 \geq j_2$. After the first local rise, $e$ goes to level $j_1+1$. Thus, it is no longer in $N_{j_2}(S_2)$, meaning overall $|N_{j_2}(S_2)|$ did not grow from the insertion and first local rise. Since $S_2$ was not PD right before the insertion of $e$ (by Invariant \ref{inv1}), it is not PD after the first local rise. We conclude that only one local rise can occur per update step, and the claim holds.
\qed 
\end{proof}

\begin{cl} \label{prK}
$m'_3 \leq K$.
\end{cl}

\begin{proof}
We will show that we can assign a separate element to each set $S \in \mathcal{F}'_3$. Since there are no more than $K$ elements (up to one inserted per update step), this will conclude the proof. The element $e$ that we assign to $S \in \mathcal{F}'_3$ is an element that joins $Cov(S)$ when it is created in a partial reset (at least one must exist). Assume by contradiction that $e$ is assigned to two sets in $\mathcal{F}'_3$, denoted by $S_1$ and $S_2$. Since $e$ joins $Cov(S_1)$ and $Cov(S_2)$ following a reset, these resets must occur at different time steps, so assume that the reset where $Cov(S_1)$ is created occurs first. After $e \in Cov(S_1)$, for it to join $Cov(S_2)$ in the second reset it must be part of $\tilde{\mathcal{U}}$ in that second reset. Since $S_1$ cannot participate in this reset in $\tilde{\mathcal{F}}$ (by definition of $\mathcal{F}'_3$), $l_{cov}(S_1)$ must be higher than $l(e)$ right before the second reset. This, however, is a contradiction to Invariant \ref{inv3}. The claim follows.
\qed
\end{proof}

\begin{cor} \label{endrc}
$m' = O(K)$.
\end{cor}

\begin{proof}
We have that $m' = m'_1 + m'_2 + m'_3$, thus by combining Claims \ref{insK}, \ref{lrK} and \ref{prK}, we get that $m' = O(K)$.
\qed
\end{proof}

\begin{cor} \label{ninth}
The amortized recourse is $O(\epsilon^{-4} \cdot \log C) = O_{\epsilon}(\log C)$.
\end{cor}

\begin{proof}
By Lemma \ref{seventh}, Observation \ref{eighthh} and Corollary \ref{endrc}, the total recourse satisfies 
$$\sum_{S \in \mathcal{F}} t(S) ~\leq~ 3 \cdot \left( \sum_{S \in \mathcal{F}} r(S)\right) + m' ~=~ 3 \cdot \left( \sum_{p=1}^q |S_p| \right) + m' = O(\epsilon^{-4} \cdot \log C \cdot K) + O(K) = O(\epsilon^{-4} \cdot \log C \cdot K),$$ 

\noindent 
yielding the required result.
\qed
\end{proof}

\noindent \textbf{Remark:} Note that for DS, Observation \ref{eighth} turns into $t(v) \leq 2r(v)$ for any $v \in V$, since in DS all vertices begin \emph{in} the dominating set, as opposed to SC where the set cover begins empty. Thus, a vertex must leave the DS first in order to join it later. Due to this difference, in DS the recourse bound follows immediately from Lemma \ref{seventh} and Observation \ref{eighth}.

\section{Approximation Factor} \label{approx}

\begin{lem} \label{lem4.5sc}
Let $OPT$ be the (dynamic) minimum cost of any collection of sets that covers the (dynamic) universe of elements $\mathcal{U}$. Recall that $\mathcal{C}$ denotes the total cost of all sets in the set cover obtained by our algorithm. Then $\mathcal{C} < (1+31\epsilon) \cdot \ln n \cdot OPT$.
\end{lem}

\begin{proof}
For any element $e \in \mathcal{U}$ covered by a set $S$ (meaning $e \in Cov(S)$), let: 

\begin{equation} \label{firstap}
q_e = \frac{1}{\beta^{l_{cov}(S)}}
\end{equation} 

\noindent Moreover, for any level $j$ observe that:

\begin{equation} \label{secondap}
\mathcal{C}_j^j = \sum_{S \in \mathcal{S}_j^j} c(S) \leq \mathcal{D}_j^j + \sum_{e \: : \: l(e)=j}q_e,
\end{equation}

\noindent where recall that $\mathcal{D}_j^j$ is the dirt counter of level $j$. Indeed, for any set $S \in \mathcal{S}_j^j$, by Observations \ref{obclean} and \ref{parclean} the initial $Cov(S)$ was of size at least $c(S) \cdot \beta^j$. Since every leaving element raises $\mathcal{D}_j^j$ by $\frac{1}{\beta^j}$, and every remaining element $e$ has $q_e = \frac{1}{\beta^j}$, we get that every element in the initial $Cov(S)$ contributes $\frac{1}{\beta^j}$ to the right hand side. Therefore, together they contribute at least $c(S)$. Summing up on all levels we get:

\begin{equation} \label{secondap1}
\mathcal{C} = \sum_{S \in \mathcal{S}} c(S) \leq \mathcal{D} + \sum_{e \in \mathcal{U}}q_e
\end{equation}  

\noindent By Invariant \ref{inv2}, $\mathcal{D} < \frac{\epsilon}{\beta} \cdot \mathcal{C}$. Plugging in Equation (\ref{secondap1}), we obtain:

\begin{equation} \label{secondap2}
\mathcal{C} < \beta \cdot \sum_{e \in \mathcal{U}}q_e
\end{equation} 

\noindent Denote by $S_1,S_2, \ldots S_k$ the sets in the optimal set cover, and consider one of these sets, $S_i$. Let $S_i = \{e_1, e_2, \ldots, e_{x_i}\}$ where $x_i = |S_i|$. Notice that each element $e_j \in S_i$ is not necessarily in $Cov(S_i)$. Without loss of generality, assume that $l(e_1) \geq l(e_2) \geq \ldots \geq l(e_{x_i})$. We make the following claim, which relies on Invariant \ref{inv1} (there are no PD sets).

\begin{cl} \label{importantap}
For any $e_j \in S_i$, $q_{e_j} \leq \frac{c(S_i)}{x_i - j + 1} \cdot \beta^3$.
\end{cl}

\begin{proof}
Assume by contradiction that $q_{e_j} > \frac{c(S_i)}{x_i - j + 1} \cdot \beta^3$, and denote $l=l(e_j)$. By Equation (\ref{firstap}), $q_{e_j} = \frac{1}{\beta^{l}}$. Thus, $\frac{1}{\beta^{l}} > \frac{c(S_i)}{x_i-j+1} \cdot \beta^3$, and we obtain:

\begin{equation} \label{thirdap}
x_i-j+1 > c(S_i) \cdot \beta^{l+3}
\end{equation}

\noindent Since $l(e_{j'}) \leq l$ for any $j' \geq j$, we know that all of the $(x_i-j+1)$ elements: $e_j, e_{j+1}, \ldots, e_{x_i}$ are covered at a level $\leq l$. Since all of these elements are in $S_i$, we know that $|N_{l+1}(S_i)| \geq x_i-j+1$. Therefore, from Equation (\ref{thirdap}) we get that $|N_{l+1}(S_i)| > c(S_i) \cdot \beta^{l+3}$. Since we maintain $|N_{l+1}(S_i)|$ with a possible additive error of $\epsilon \cdot \beta^{l+1}$ (see Section \ref{updatetime} for details), we get that $|N_{l+1}(S_i)| > c(S_i) \cdot \beta^{l+2}$. Hence, $S_i$ is $(l+1)$-PD by definition, which is a contradiction to Invariant \ref{inv1}. 
\qed
\end{proof}

\noindent Equipped with Claim \ref{importantap}, we continue now with the proof of Lemma \ref{lem4.5sc}. Summing up on all $e \in S_i$, we get that:

\begin{equation} \label{fourthap}
\sum_{e \in S_i} q_e \leq \beta^3 \cdot c(S_i) \cdot (1+\frac{1}{2}+\frac{1}{3}+ \ldots + \frac{1}{x_i}) \leq \beta^3 \cdot c(S_i) \cdot H_{n} \leq \beta^3 \cdot c(S_i) \cdot (\ln n +1),
\end{equation}

\noindent where $H_{n}$ is the $n$th harmonic number. Now, by summing up on all $k$ sets $S_i$ in the optimal solution, we obtain:

\begin{equation} \label{fifthap}
\sum_{i=1}^k \sum_{e \in S_i} q_e \leq \sum_{i=1}^k \beta^3 \cdot c(S_i) \cdot (\ln n +1) = OPT \cdot \beta^3 \cdot (\ln n +1)
\end{equation}

\noindent Finally, since the MSC covers all elements, we know that:

\begin{equation} \label{sixthap}
\sum_{e \in \mathcal{U}}q_e \leq \sum_{i=1}^k \sum_{e \in S_i} q_e
\end{equation}

\noindent Therefore, by Equations (\ref{secondap2}), (\ref{fifthap}) and (\ref{sixthap}), we obtain:

\begin{equation} \label{seventhap}
\mathcal{C} < OPT \cdot \beta^4 \cdot (\ln n +1)
\end{equation}

\noindent We shall assume $\ln n \geq \frac{1}{\epsilon}$, thus $\mathcal{C} < OPT \cdot \beta^5 \cdot \ln n$. For $\epsilon<1$, we get that $\beta^5 < 1+31 \epsilon$, thus we finally obtain:

$$\mathcal{C} < OPT \cdot (1+31\epsilon) \cdot \ln n$$

\qed
\end{proof}

\noindent Since $\mathcal{C}$ is the cost of the SC given by the algorithm and $OPT$ is the cost of the optimal SC, with a simple scaling argument we obtain the approximation ratio of $(1+\epsilon) \cdot \ln n$.

\bigskip

\noindent \textbf{Remark.} If the size of each set is upper bounded by $\kappa$, then in Equation (\ref{fourthap}) we get $H_{\kappa}$, and so we actually obtain $\mathcal{C} < OPT \cdot (1+31\epsilon) \cdot \ln \kappa$. In the particular case of DS, since each vertex can dominate up to $\Delta+1$ vertices, we get $\mathcal{C} < OPT \cdot (1+31\epsilon) \cdot \ln (\Delta+1)$.